\documentclass{elsarticle}

\usepackage{lineno,hyperref}
\modulolinenumbers[5]

\usepackage{graphicx}
\usepackage{enumitem}
\usepackage{pgf,tikz}
\usepackage[ruled,vlined]{algorithm2e} 
\SetAlFnt{\small}
\SetAlCapFnt{\small}
\SetAlCapNameFnt{\small}
\SetAlCapHSkip{0pt}
\IncMargin{-\parindent}

\usepackage{amsmath}
\usepackage{amsthm}
\usepackage{amsfonts}
\usepackage{multirow}
\usepackage{lscape}
\usepackage{afterpage}
\usepackage{subcaption}
\usepackage{soul}

\allowdisplaybreaks

\newtheorem{theorem}{Theorem}
\newtheorem{proposition}{Proposition}
\newtheorem{corollary}{Corollary}

\newtheorem{definition}{Definition}

\newcommand{\bigtimes}{\text{\LARGE $\times$}}
\newcommand{\NPHard}{$\mathsf{NP}$-hard}
\newcommand{\NP}{$\mathsf{NP}$}
\newcommand{\Poly}{$\mathsf{P}$}
\newcommand{\PolyAPX}{Poly-$\mathsf{APX}$}

\journal{Journal of \LaTeX\ Templates}

\begin{document}

\begin{frontmatter}

\title{Leadership in Singleton Congestion Games: \protect\\ What is Hard and What is Easy}
\tnotetext[mytitlenote]{A preliminary version of this work appeared in~\cite{Marchesi18:leadership}.
  This extended version contains new computational complexity proofs for problems addressed only partially in the preliminary version, proofs of results concerning problems not addressed previously, new algorithms, and their experimental evaluation.}

\author{Matteo Castiglioni, Alberto Marchesi, Nicola Gatti}
\address{Politecnico di Milano, Italy}

\author{Stefano Coniglio}
\address{University of Southampton, UK}





\begin{abstract}
We study the problem of computing \emph{Stackelberg equilibria} \emph{Stackelberg games} whose underlying structure is in {\em congestion games}, focusing on the case where each player can choose a single resource (a.k.a. \emph{singleton congestion games}) and one of them acts as leader. In particular, we address the cases where the players either have the same action spaces (i.e., the set of resources they can choose is the same for all of them) or different ones, and where their costs are either monotonic functions of the resource congestion or not.

We show that, in the case where the players have different action spaces, the cost the leader incurs in a Stackelberg equilibrium cannot be approximated in polynomial time up to within any polynomial factor in the size of the game unless $\mathsf{P}=\mathsf{NP}$, independently of the cost functions being monotonic or not. We show that a similar result also holds when the players have nonmonotonic cost functions, even if their action spaces are the same. Differently, we prove that the case with identical action spaces and monotonic cost functions is easy, and propose polynomial-time algorithm for it.
We also improve an algorithm for the computation of a socially optimal equilibrium in singleton congestion games with the same action spaces without leadership, and extend it to the computation of a Stackelberg equilibrium for the case where the leader is restricted to pure strategies.

For the cases in which the problem of finding an equilibrium is hard, we show how, in the optimistic setting where the followers break ties in favor of the leader, the problem can be formulated via mixed-integer linear programming techniques, which computational experiments show to scale quite well.
\end{abstract}


\end{frontmatter}


\section{Introduction}\label{sec:introduction}

In {\em Stackelberg Games} (SGs), a player, acting as \emph{leader}, has the ability to commit to a (possibly) mixed strategy beforehand, while the other players, acting as \emph{followers}, observe the leader's commitment and, then, decide how to play~\cite{von2010leadership}. 
Over the last years, Stackelberg games and the corresponding \emph{Stackelberg Equilibria} (SEs) have received a lot of attention in the artificial intelligence literature.
The vast majority of the works related to the topic has focused on the problem of finding an SE when mixed-strategy commitments are allowed---a problem often referred to as \emph{computing the optimal strategy to commit to}~\cite{conitzer2006computing}.
%
%
%
%

The recent surge of interest in SGs is mainly motivated by their many successful real-world applications.
A prominent one is that of \emph{security games}, which model situations where a defender, acting as leader, has to allocate scarce resources to protect valuable targets from an attacker, acting as follower~\cite{paruchuri2008playing,KiekintveldJTPOT09,an2011guards,tambe2011security}.
Other interesting application are found in {\em toll-setting games}, where the leader is a central authority collecting tolls from the users of a network who, acting as followers, decide on how to best travel through the network so to minimize their cost after observing the pricing strategy chosen by the authority~\cite{labbe1998bilevel,labbe2016bilevel}.
Besides the security domain and toll-setting games, applications of SGs can be found in, among others, interdiction games~\cite{caprara2016bilevel,matuschke2017protection}, network routing~\cite{amaldi2013network}, and mechanism design~\cite{sandholm2002evolutionary}.

\subsection{State of the Art on Stackelberg Equilibrium Computation}

When studying SGs, two crucial aspects need to be considered: how the followers break \emph{ties} among the multiple equilibria that could arise after observing the leader's commitment, and the \emph{structure} of the underlying followers' game.

As to the first aspect, 
%
two extreme cases are usually taken into account: the optimistic and the pessimistic one. 
In an \emph{Optimistic Stackelberg Equilibrium} (OSE), the followers are assumed to break ties in favor of the leader.
In a \emph{Pessimistic Stackelberg Equilibrium} (PSE), they are assumed to do it against her~\cite{von2010leadership}.\footnote{Many works refer to OSEs and PSEs as, respectively, \emph{strong} SEs and \emph{weak} SEs, following the terminology of~\cite{Breton88:Sequential}, where the two concepts are first introduced.}

As to the second aspect, Table~\ref{tab:state_of_art} summarizes most of the known computational results for the problem of computing an O/PSE according to the structure of the underlying game.
The problem is known to be easy in 2-player normal-form games in both the optimistic and pessimistic setting, as shown in, respectively,~\cite{conitzer2006computing}~and~\cite{von2010leadership}.
In particular, \cite{conitzer2011commitment} shows that the problem of computing an OSE can be formulated as a single \emph{Linear Program} (LP), while~\cite{von2010leadership} illustrates that a PSE can be computed by solving a polynomial number of LPs.

When one considers the case of $n$-player normal-form games with $n \geq 3$, many cases are possible, depending on how the followers behave after observing the leader's commitment.
A reasonable choice, which has been widely investigated in the literature, is to assume that they play simultaneously and noncooperatively, reaching a \emph{Nash Equilibrium} (NE).
In this case, it is known that finding an O/PSE is not in \PolyAPX~unless \Poly $=$ \NP, even when there are only two followers (i.e., with $n = 3$)~\cite{basilico2017methods}.
Computing an OSE becomes easy for any $n$ if we restrict the followers to only play pure strategies, as it requires the solution of an LP for each outcome of the followers' game, whose number is polynomial in the size of the game representation~\cite{algo2018computing}.
On the other hand, computing a PSE is still \NPHard~even with only two followers playing pure strategies~\cite{coniglio2017pessimistic} and it is not in \PolyAPX~unless \Poly $=$ \NP~with at least three followers (i.e., with $n \geq 4$)~\cite{algo2018computing}.

As for algorithms, the authors of~\cite{basilico2017bilevel} show how to formulate the problem of finding an OSE in $n$-player normal-form games as a \emph{nonlinear} and \emph{nonconvex} mathematical program, which they solve via spatial branch-and-bound techniques.
As shown in~\cite{coniglio2017pessimistic}, when the followers are restricted to play pure strategies a PSE can be found by employing an algorithm that solves multiple \emph{lexicographic Mixed-Integer Linear Programs} (lex-MILPs), which, as the authors show, can be further enhanced by embedding it in a branch-and-bound scheme.

As to works on $n$-player normal-form games where the followers do not play an NE, \cite{conitzer2011commitment} shows that finding an OSE is easy when the followers can play correlated strategies, while~\cite{conitzer2006computing} proves that the problem becomes \NPHard~if the followers play in a hierarchical fashion.

Besides normal-form games, the literature has devoted considerable attention to Bayesian 2-player normal-form games where the follower can be of different types, mainly due to their relevance in security games.
In this setting, it is known that finding an OSE is \PolyAPX-complete~\cite{letchford2009learning}, and that an equilibrium can be found by solving a \emph{Mixed-Integer Linear Program} (MILP)~\cite{paruchuri2008playing}.
As recently shown in~\cite{alberto2018computing}, the same hardness result also holds for the problem of computing a PSE. \cite{alberto2018computing} also provides an algorithm for computing an equilibrium via the solution of exponentially many LPs, without resorting to the normal-form representation.
%

Over the last years, the Stackelberg paradigm has also been applied to 2-player extensive-form games.
In particular, the authors of~\cite{Letchford10:Computing} prove that finding an OSE is \NPHard~even in games without nature.
The results are extended by the authors of~\cite{Farina18:trembling}, who prove that computing a PSE is also \NPHard.
%
%
%
Works such as~\cite{Bosansky15:Sequence,Cermak16:Using} address the problem of computing an OSE in extensive-form games, providing worst-case exponential time algorithms based on MILPs.
In the context of extensive-form games, attempts have also been made towards the refinement of SEs.
In particular, the authors of~\cite{Kroer18:Robust} introduce the idea of a \emph{robust} SEs, where an optimal commitment is found against a worst-case follower's utility model.
Pursuing a different approach, the authors of~\cite{Farina18:trembling} show how to guarantee an optimal commitment off the equilibrium path by adopting the idea of \emph{trembling-hand perfection} to the Stackelberg setting.

Other works attempted to relax the general structure of normal-form games, trying to identify games with many players where SEs can be efficiently computed.
Along this line of research, the authors of~\cite{alberto2018computing} analyze \emph{polymatrix games}, which are games where the players interact pairwise and each player takes part to a 2-player normal-form game with each of the other players.
For these games, \cite{basilico2017methods} shows that, when the followers play mixed strategies, finding an O/PSE is not in \PolyAPX~unless \Poly $=$ \NP. The result is extended in~\cite{alberto2018computing}, where the authors show it to hold even when the followers are restricted to pure strategies.
While, for fixed $n$, finding an OSE with the followers playing pure strategies is easy, the same does not hold for PSEs, as the problem is hard even with only three followers (i.e., with $n = 4$)~\cite{de2018computing}.

For works applying the Stackelberg paradigm to other game models, such as stochastic games and Bayesian signaling games, we refer the reader to~\cite{letchford2012computing,vorobeychik2012computing,xu2016signaling}.

%
%
%
%

\afterpage{
\begin{landscape}
  \begin{table}[!htp]
    \caption{Summary of the results on the computation of SEs in normal-form games, Bayesian normal-form games, extensive-form games, and polymatrix games. In all cases, we assume that the leader is allowed to play mixed strategies.}
    \label{tab:state_of_art}
    \centering
    \renewcommand{\arraystretch}{1.3}
    \setlength{\tabcolsep}{2pt}
    \begin{tabular}{c|l||c|c||c|c}
      \multicolumn{2}{r||}{} & \multicolumn{2}{c||}{Optimistic} & \multicolumn{2}{c}{Pessimistic} \\
      \hline
      \multicolumn{2}{r||}{Followers' strategies} & Pure & Mixed & Pure & Mixed \\
      \hline
      \hline
      \multicolumn{6}{c}{Normal-form games} \\
      \hline
      \hline
      \multirow{2}{*}{$n=2$} & Complexity & \multicolumn{2}{|c||}{ \Poly~\cite{conitzer2006computing}} & \multicolumn{2}{|c}{\Poly~\cite{von2010leadership}}   \\
      & Algorithm & \multicolumn{2}{|c||}{multi-LP~\cite{conitzer2006computing}} & \multicolumn{2}{|c}{multi-LP~\cite{von2010leadership}} \\
      \hline
      \multirow{2}{*}{$n = 3$} & Complexity & \Poly~\cite{algo2018computing} &  \NPHard, $\notin$ \PolyAPX~\cite{basilico2017methods} &  \NPHard~\cite{coniglio2017pessimistic} & \NPHard, $\notin$ \PolyAPX~\cite{basilico2017methods} \\
      & Algorithm & multi-LP~\cite{algo2018computing} & spatial branch-and-bound~\cite{basilico2017bilevel}  &  multi-lex-MILP~\cite{coniglio2017pessimistic} & -- \\
      \hline
      \multirow{2}{*}{$n\geq 4$} & Complexity & \Poly~\cite{algo2018computing} &  \NPHard, $\notin$ \PolyAPX~\cite{basilico2017methods} &  \NPHard~\cite{coniglio2017pessimistic}, $\notin$ \PolyAPX~\cite{algo2018computing} & \NPHard, $\notin$ \PolyAPX~\cite{basilico2017methods} \\ 
      & Algorithm & multi-LP~\cite{algo2018computing} & spatial branch-and-bound~\cite{basilico2017bilevel}  &  multi-lex-MILP~\cite{coniglio2017pessimistic}  & -- \\
      \hline
      \hline
      \multicolumn{6}{c}{Bayesian normal-form games} \\
      \hline
      \hline
      \multirow{2}{*}{$n=2$} & Complexity & \multicolumn{2}{|c||}{\NPHard~\cite{conitzer2006computing},  \PolyAPX-complete~\cite{letchford2009learning}} & \multicolumn{2}{|c}{\NPHard,  \PolyAPX-complete~\cite{alberto2018computing}} \\
      & Algorithm & \multicolumn{2}{|c||}{MILP~\cite{paruchuri2008playing}} & \multicolumn{2}{|c}{multi-LP~\cite{alberto2018computing}} \\
      \hline
      \hline
      \multicolumn{6}{c}{Extensive-form games} \\
      \hline
      \hline
      \multirow{2}{*}{$n=2$} & Complexity & \multicolumn{2}{|c||}{\NPHard~\cite{letchford2012computing}} & \multicolumn{2}{|c}{\NPHard~\cite{Farina18:trembling}}  \\
      & Algorithm & \multicolumn{2}{|c||}{MILP~\cite{Bosansky15:Sequence,Cermak16:Using}} & \multicolumn{2}{|c}{multi-LP~\cite{von2010leadership}}  \\
      \hline
      \hline
      \multicolumn{6}{c}{Polymatrix games} \\
      \hline
      \hline
      \multirow{2}{*}{$n= 3$} & Complexity & \Poly~\cite{algo2018computing} & \NPHard, $\notin$ \PolyAPX~\cite{basilico2017methods} & \NPHard & \NPHard, $\notin$ \PolyAPX~\cite{basilico2017methods}  \\
      & Algorithm & multi-LP~\cite{algo2018computing} & spatial branch-and-bound~\cite{basilico2017bilevel} & multi-lex-MILP~\cite{algo2018computing} & -- \\
      \hline
      \multirow{2}{*}{$n\geq 4$ (fixed)} & Complexity & \Poly~\cite{algo2018computing} & \NPHard, $\notin$ \PolyAPX~\cite{basilico2017methods} & \NPHard, $\notin$ \PolyAPX~\cite{de2018computing}
      & \NPHard, $\notin$ \PolyAPX~\cite{basilico2017methods}  \\
      & Algorithm & multi-LP~\cite{algo2018computing} & spatial branch-and-bound~\cite{basilico2017bilevel} & multi-lex-MILP~\cite{algo2018computing} & -- \\\hline
      \multirow{2}{*}{$n \geq 4$ (free)} & Complexity & \NPHard, $\notin$ \PolyAPX~\cite{alberto2018computing} & \NPHard, $\notin$ \PolyAPX~\cite{basilico2017methods} & \NPHard, $\notin$ \PolyAPX~\cite{de2018computing} & \NPHard, $\notin$ \PolyAPX~\cite{basilico2017methods}  \\
      & Algorithm & multi-LP~\cite{algo2018computing} & spatial branch-and-bound~\cite{basilico2017bilevel} & multi-lex-MILP~\cite{algo2018computing} & -- \\
    \end{tabular}
  \end{table}
\end{landscape}}

\subsection{The Stackelberg Paradigm in Congestion Games}

We focus, in this work, on \emph{Congestion Games} (CGs),
which
model situations in which the players compete for the use of a finite set of shared resources.
The players' actions are subsets of the resources and the costs the players incur depend (monotonically or not) on the level of resource utilization, typically referred to as \emph{resource congestion}.
%
%
Crucially, CGs always admit pure-strategy NEs~\cite{rosenthal1973class}.
Such equilibria are always achievable by \emph{best-response dynamics}, i.e., by applying an iterative procedure by which, at each iteration, a player changes her action and switches to playing a best-response to the actions currently played by the other players~\cite{monderer1996potential}.

Many classes of CGs have been introduced in the literature.
These games can be characterized according to the combinatorial structure of the players' action spaces.
In this work, we focus on \emph{Singleton CGs} (SCGs)~\cite{ieong2005fast}, i.e., on CGs where each player cannot use more than a single resource.
Computing NEs in SCGs is easy~\cite{ackermann2008impact} and, for the case in which all the players have the same action space (we will refer to these games as \emph{symmetric}), finding an NE minimizing the social cost is also easy~\cite{ieong2005fast}.

Other classes of CGs have been studied in the literature.
For instance, the authors of~\cite{ackermann2008impact} propose a generalization of SCGs where a player's action space is expressed as a \emph{matroid} defined over the set of resources.
Many works have also addressed CGs played on a network, e.g., games where the players' actions are paths connecting a source to a destination~\cite{fabrikant2004complexity}, or spanning trees~\cite{werneck2000finding}.

In this work, we apply a Stackelberg paradigm to SCGs, assuming the presence of a special player acting as leader.
%
%
The leader commits to a (possibly) mixed strategy, while all the other players, acting as followers, observe the leader's commitment and then decide how to play,
 reaching an NE in the resulting SCGs.
%
%
In particular, we study the case in which the followers play pure strategies after observing the leader's commitment, which is reasonable as this followers' game always admits at least a pure-strategy NE reachable by best-response dynamics.
%
%
For the sake of generality, we assume that the leader's cost may differ from the followers'.

Practical scenarios where Stackelberg SCGs are relevant are those where the set of players contains a higher-priority player who can decide which resource to use before the other ones do.
%

While the Stackelberg paradigm has already been applied to CGs, the only works which, to our knowledge, pursue this line of research are~\cite{roughgarden2004stackelberg} and its extensions~\cite{fotakis2010stackelberg,bonifaci2010stackelberg,bilo2015stackelberg}. 
We remark, though, that the author of~\cite{roughgarden2004stackelberg} considers a different Stackelberg paradigm where the leader is an authority whose objective is to minimize the social cost of the NE reached by the followers.
Differently, in this work we assume that the leader is a special player who has the ability to commit to a strategy beforehand with the aim of minimizing her own cost.

\subsection{Original Contributions}

We provide, in this work, an extensive study of the problem of computing SEs in SCGs with leadership.
In particular, we identify four possible cases, according to two orthogonal features of SCGs. 
The first one concerns the relationship among the action spaces of the players. We analyze two possibilities: the one where the players share the same set of resources, and the one where the sets of resources available to them may differ.
The second feature we address is related to the shape of the players' cost functions. We consider two cases: the one where these functions are monotonically increasing in the resource congestion, and the one in which they are not.

\begin{table}[!htp]
  \caption{Summary of the original contributions on the problem of computing an O/PSEs in SCGs provided in this paper.}
  \label{tab:original}
  \centering
  \renewcommand{\arraystretch}{1.2}
  \setlength{\tabcolsep}{2pt}
  \begin{tabular}{c|c|l||c|c}
    \multicolumn{5}{c}{Optimistic}  \\
    \hline
    \multicolumn{3}{r||}{Leader's commitment} & Pure & Mixed \\
    \hline
    \multirow{4}{*}{ {\renewcommand{\arraystretch}{0.7}\begin{tabular}{@{}c@{}}Identical \\ action spaces \\(symmetric)\end{tabular}} } & \multirow{2}{*}{  {\renewcommand{\arraystretch}{0.7}\begin{tabular}{@{}c@{}}Monotonic \\ costs\end{tabular}}  } & Complexity & \Poly & \Poly \\
    & & Algorithm & Greedy & Greedy \\
    \cline{2-5}
    & \multirow{2}{*}{  {\renewcommand{\arraystretch}{0.7}\begin{tabular}{@{}c@{}}Non-monotonic \\ costs\end{tabular}}   } & Complexity & \Poly &  \NPHard, $\notin$ \PolyAPX  \\
    & & Algorithm & Dynamic Programming & MILP  \\
    \hline
    \multirow{4}{*}{  {\renewcommand{\arraystretch}{0.7}\begin{tabular}{@{}c@{}}Different \\ action spaces\end{tabular}} } & \multirow{2}{*}{  {\renewcommand{\arraystretch}{0.7}\begin{tabular}{@{}c@{}}Monotonic \\ costs\end{tabular}}  } & Complexity & \NPHard, $\notin$ \PolyAPX & \NPHard, $\notin$ \PolyAPX \\
    & & Algorithm & MILP & MILP \\
    \cline{2-5}
    & \multirow{2}{*}{   {\renewcommand{\arraystretch}{0.7}\begin{tabular}{@{}c@{}}Non-monotonic \\ costs\end{tabular}} } & Complexity & \NPHard, $\notin$ \PolyAPX & \NPHard, $\notin$ \PolyAPX \\
    & & Algorithm & MILP & MILP \\
    \hline\hline
    \multicolumn{5}{c}{Pessimistic}  \\
    \hline
    \multicolumn{3}{r||}{Leader's commitment} & Pure & Mixed \\
    \hline
    \multirow{4}{*}{ {\renewcommand{\arraystretch}{0.7}\begin{tabular}{@{}c@{}}Identical \\ action spaces\\(symmetric)\end{tabular}}  } & \multirow{2}{*}{ {\renewcommand{\arraystretch}{0.7}\begin{tabular}{@{}c@{}}Monotonic \\ costs\end{tabular}} } & Complexity & \Poly & \Poly \\
    & & Algorithm & Greedy & Greedy \\
    \cline{2-5}
    & \multirow{2}{*}{ {\renewcommand{\arraystretch}{0.7}\begin{tabular}{@{}c@{}}Non-monotonic \\ costs\end{tabular}}   } & Complexity & \Poly &  \NPHard, $\notin$ \PolyAPX  \\
    & & Algorithm & Dynamic Programming & --  \\
    \hline
    \multirow{4}{*}{ {\renewcommand{\arraystretch}{0.7}\begin{tabular}{@{}c@{}}Different \\ action spaces\end{tabular}}  } & \multirow{2}{*}{  {\renewcommand{\arraystretch}{0.7}\begin{tabular}{@{}c@{}}Monotonic \\ costs\end{tabular}} } & Complexity & \NPHard, $\notin$ \PolyAPX & \NPHard, $\notin$ \PolyAPX \\
    & & Algorithm & -- & -- \\
    \cline{2-5}
    & \multirow{2}{*}{ {\renewcommand{\arraystretch}{0.7}\begin{tabular}{@{}c@{}}Non-monotonic \\ costs\end{tabular}}  } & Complexity & \NPHard, $\notin$ \PolyAPX & \NPHard, $\notin$ \PolyAPX \\
    & & Algorithm & -- & -- \\
  \end{tabular}
\end{table}

Table~\ref{tab:original} summarizes the original results that we provide with this work.
In particular, we show that, in SCGs where the players' action spaces can be different, computing an O/PSE is not in \PolyAPX~unless \Poly $=$ \NP, even when the players' cost functions are monotonic, the leader has only one available action, and her costs are equal to the followers'.
This also shows that, as we will better explain in the following, the same inapproximability result also holds for the problem of computing, in the same game setting, an NE which minimizes/maximizes the cost incurred by any given player.

For the symmetric case where the players have access to the same set of resources, we show that the complexity of computing an O/PSE depends on the nature of the players' cost functions. 
We prove that the problem is not in \PolyAPX~unless \Poly $=$ \NP~for the case where the players' costs are nonmonotonic functions of the resource congestion.
On the other hand, we show that, in the symmetric case where the players have access to the same set of resources, the problem of computing an O/PSE can be solved in polynomial time when the cost functions are monotonic.
%
%
While proving the correctness of the algorithm is straightforward when the leader's commitment is a pure strategy, the analysis is more involved with mixed-strategy commitments.
Our result follows, as we will show, from the fact that mixed-strategy commitments do not allow the leader to incur a cost smaller than the one she incurs with a pure-strategy.
We also consider the case where the leader is restricted to pure-strategy commitments, providing a polynomial-time algorithm for its solution.

Finally, we provide two mathematical programming formulations to compute an OSE for
%
games with different action spaces in, at most, exponential time, and a more compact one for the symmetric case.
We also evaluate, experimentally, the scalability of the two formulations when fed to a state-of-the-art MILP solver
and compare their performance---in terms of computing time and solution efficiency---to simple algorithms based on the repetition of best response dynamics.
%


\subsection{Structure of the Work}

The remainder of the paper is organized as follows.
Section~\ref{sec:preliminaries} introduces basic concepts and the notation we use, including the formal definitions of the game models we consider.
Section~\ref{sec:complexity_different_actions} provides the main hardness results for the problem of computing an O/PSE in games where the players' action spaces are different.
Section~\ref{sec:complexity_same_actions} does the same in games where the players' action spaces are all equal but the cost functions are nonmonotonic.
Section~\ref{sec:algorithms_same_actions_polynomial} establishes which problems can be solved efficiently, providing the corresponding polynomial-time algorithms.
Section~\ref{sec:practical_algorithms} proposes
mathematical programming formulations for computing an OSE in the intractable cases and assesses their scalability via computational experiments.
%
%
Finally, Section~\ref{sec:conclusions} concludes the work summarizing the results and pointing out directions for future research.

\section{Preliminaries}\label{sec:preliminaries}

In this work, we analyze SCGs in which a \emph{leader} commits to a strategy beforehand, and, then, the \emph{followers} simultaneously decide how to play, reaching an NE in the game that results from observing the leader's commitment.
Adopting the notation introduced in~\cite{shoham2008multiagent}, we provide the following formal definition of the class of games we study:
\begin{definition}[Stackelberg SCG (SSCG)]\label{def:sscg}
	A \emph{Stackelberg SCG (SSCG)} is a tuple $(N,R,A,c_\ell,c_f)$, where:
	\begin{itemize}
		\item $N=F \cup \{ \ell \}$ is a finite set of players, $\ell$ being the leader and $F$ the set of followers;
		\item $R$ is a finite set of resources;
		\item $A= \{ A_p \}_{p \in N}$ is the set of all players' actions, with $A_p$, for each $p \in N$, being the set of actions of player $p$;
		\item $c_\ell = \{c_{i,\ell}\}_{i \in R}$ and $c_f = \{c_{i,f}\}_{i \in R}$ are, respectively, the leader's and followers' cost functions, with $c_{i,\ell}, c_{i,f} : \mathbb{N} \rightarrow \mathbb{Q}$ being the costs of resource $i$ as a function of its congestion for, respectively, the leader and the followers.
	\end{itemize} 
\end{definition}
%
%
%
%
%
%
We denote by $n$ and $r$ the number of players and of resources (i.e., $n:= |N|$ and $r := |R|$.
As usual, we assume $c_{i,\ell}(0) = c_{i,f}(0) = 0$ for every $i \in R$.
%
%

%
We call the players' cost functions \emph{weakly monotonic} if, for every resource $i \in R$, $c_{i,\ell}(x) \leq c_{i, \ell}(x+1)$ and $c_{i,f}(x) \leq c_{i, f}(x+1)$ for all $x \in \mathbb{N}$, and {\em strictly monotonic} if all the inequalities are strict.
%

We call \emph{strategy} of player $p \in N$ a probability distribution $\sigma_p$ over $A_p$, where $\sigma_p(a_p)$ denotes the probability that $a_p \in A_p$ is played. 
Let $\Delta_p$ be the set of player $p$'s strategies.
A strategy $\sigma_p \in \Delta_p$ is said \emph{pure} if it prescribes player $p$ to always play some action $a_p \in A_p$, i.e., if $\sigma_p(a_p) = 1$ and $\sigma_p(a_p') = 0$ for all $a_p' \in A_p \setminus\{a_p\}$. Otherwise, $\sigma_p$ is said \emph{mixed}.
A collection of strategies is called \emph{strategy profile} in general, and \emph{action profile} if all the strategies it contains are pure.
In this work, we collectively denote by $\sigma = (\sigma_\ell, a)$ a strategy profile in which the leader plays a (possibly) mixed strategy $\sigma_\ell \in \Delta_\ell$ and the followers play the pure strategies contained in the action profile $a = (a_p)_{p \in F} \in \bigtimes_{p \in F} A_p$.

Let $a = (a_p)_{p \in F} \in \bigtimes_{p \in F} A_p$ be a followers' action profile.
We let $\nu_i^a = | \{ p \in F \mid a_p = i \} |$ be the number of followers selecting resource $i \in R$ in $a$. This quantity is equal to the resource congestion caused by the followers' presence only.
We call \emph{followers' configuration} (induced by action profile $a$) the vector $\nu^a \in \mathbb{N}^r$ whose $i$-th component is $\nu_i^a$ for all $i \in R$.

For any $\sigma_\ell \in \Delta_\ell$, we define the followers' expected cost for resource $i \in R$ given $\sigma_\ell$ as the function $c_{i,f}^{\sigma_\ell} : \mathbb{N} \rightarrow \mathbb{Q}$.
$c_{i,f}^{\sigma_\ell}$ is a function of the number $x \in \mathbb{N}$ of followers who select resource $i$. Namely:
\begin{equation}
c_{i,f}^{\sigma_\ell}(x) = \sigma_\ell(i) c_{i,f}(x+1) + (1 - \sigma_\ell(i)) c_{i,f}(x).
\end{equation}
Note that, given a leader's strategy $\sigma_\ell$ and a followers' congestion $x$, all the followers who select resource $i \in R$ experience a congestion that may (with probability $\sigma_\ell(i)$) or may not (with probability $1-\sigma_\ell(i)$) be incremented by one w.r.t. $x$, depending on whether the leader chooses resource $i$ or not.
Given the strategy profile $\sigma = (\sigma_\ell, a)$, the leader's cost is:
\begin{equation}
c_{\ell}^{\sigma} = \sum_{i \in A_\ell} \sigma_\ell(i) c_{i,\ell}(\nu_i^a+1).
\end{equation}

After observing a leader's committment $\sigma_\ell$, the followers play an SCG where the resource costs are specified by the functions $c_{i,f}^{\sigma_\ell}$, for $i \in R$.
%
%
We assume that, after witnessing the leader's committment, the followers play a pure-strategy NE, which is always possible as, being a CG, the new SCG always admits one~\cite{rosenthal1973class}.

Given a strategy profile $\sigma=(\sigma_\ell,a)$, $a$ is an NE for $\sigma_\ell$ if, for every $p \in F$ and $a_p' \in A_p$, $c_{a_p,f}^{\sigma_\ell}(\nu_{a_p}^a) \leq c_{a_p',f}^{\sigma_\ell}(\nu_{a_p'}^a + 1)$,
i.e., if no follower has an incentive to unilaterally deviate from $a_p$ by selecting another resource~$a_p'$. 
For any given $\sigma_\ell \in \Delta_\ell$, let $E^{\sigma_\ell}$ be the set of NEs in the followers' game resulting from $\sigma_\ell$.

We also consider \emph{symmetric SSCGs} (SSSCGs), a subclass of SSCGs in which every player can select every resource, i.e., where $A_p = R$ for all $p \in N$:
\begin{definition}[Symmetric SSCG (SSSCG)]\label{def:ssscg}
  We call an SSCG defined by a tuple $(N,R,A,c_\ell,c_f)$ \emph{symmetric} if $A = \bigtimes_{p \in F} R$.
\end{definition}
Note that, in an SSSCG, all the followers are identical due to being allowed to choose the same resources. Thus, only the number of followers selecting each resource is significant, and a followers' action profile $a$ can be equivalently represented by the followers' configuration $\nu^a$ it induces.
As a consequence, when studying SSSCGs we do not explicitly refer to the followers' action profiles but, rather, to their configurations $\nu \in \mathbb{N}^r$, with $\sum_{i \in R} \nu_i = n-1$.
When working, rather than with action profiles, with followers' configurations, we have that $\nu$ is an NE for a given leader's strategy $\sigma_\ell \in \Delta_\ell$ if, for every $i \in R : \nu_i > 0$ and $j \in R$, $c_{i,f}^{\sigma_\ell}(\nu_i) \leq c_{j,f}^{\sigma_\ell}(\nu_j + 1)$.

Given a leader's strategy, the followers' SCG may admit multiple NEs.
%
%
As customary in the literature, we consider two extreme cases, which lead to the definition of \emph{Optimistic SE} (OSE) and \emph{Pessimistic SE} (PSE)~\cite{von2010leadership}.
In the first one, we assume that the followers act in favor of the leader, playing an NE minimizing her cost. In the second one, we assume that the followers act against the leader, always playing an NE which results in the maximum leader's cost.
As a result, OSEs and PSEs define the range of possible leader's costs over the set of SEs in the game.
Formally:
\begin{definition}\label{def:opt}
	A strategy profile $\sigma = (\sigma_\ell,a)$ is an OSE if it solves the following bilevel programming problem:
	$$
	\min_{\sigma_\ell \in \Delta_\ell} \ \ \min_{a \in E^{\sigma_\ell}} \ \ c_\ell^{(\sigma_\ell,a)}.
	$$
\end{definition}
As it is clear, an OSE always exists in SSCGs and, since the same objective function is minimized in both levels, the problem can be equivalently rewritten as:
	$$
	\min_{\substack{\sigma_\ell \in \Delta_\ell\\a \in E^{\sigma_\ell}}} \ c_\ell^{(\sigma_\ell,a)}.
	$$
	
\begin{definition}\label{def:pes}
	A PSE, if it exists, is a strategy profile $\sigma = (\sigma_\ell,a)$ which solves the following bilevel problem:
	$$
	\min_{\sigma_\ell \in \Delta_\ell} \ \ \max_{a \in E^{\sigma_\ell}} \ \ c_\ell^{(\sigma_\ell,a)}.
	$$
\end{definition}
Let us recall that, in general, the problem in Definition~\ref{def:pes} may not admit a minimum (but only an infimum) and, thus, a PSE may not exist~\cite{von2010leadership}.

\begin{proposition}
	There are SSSCGs in which a PSE does not exist.
\end{proposition}

\begin{proof}
	Consider the following instance of an SSSCGs, where $|F|=1$ and $R=\{r_1, r_2\}$.
	\begin{center}
		\renewcommand{\arraystretch}{1}\setlength{\tabcolsep}{3pt}
		\begin{tabular}{c|cc|cc}
			\hline
			$x$ & $c_{r_1,\ell}$ & $c_{r_1,f}$  & $c_{r_2,\ell}$ & $c_{r_2,f}$ \\
			\cline{1-5}
			$ 1 $ & $2$ & $1$ & $2$ & $1$ \\
			$ 2 $ & $0$ & $2$ & $2$ & $2$ \\
		\end{tabular}
	\end{center}
	Clearly, the single follower selects $r_1$ if $\sigma_\ell(r_1) < \frac{1}{2}$, she chooses $r_2$ if $\sigma_\ell(r_1) > \frac{1}{2}$, and she is indifferent between $r_1$ and $r_2$ if $\sigma_\ell(r_1) = \frac{1}{2}$.
	Thus, the leader's cost is $2 - 2 \sigma_\ell(r_1)$ if $\sigma_\ell(r_1) < \frac{1}{2}$, while it is $2$ if $\sigma_\ell(r_1) \geq \frac{1}{2}$, since, given the pessimistic assumption, the follower selects $r_2$ rather than $r_1$ when $\sigma_\ell(r_1) = \frac{1}{2}$.
	As a result, the problem in Definition~\ref{def:pes} achieves an infimum with value $1$ at $\sigma_\ell(r_1) = \frac{1}{2}$, but it does not admit a minimum.
	As a consequence, the game does not admit a PSE.
\end{proof}

\section{SSCGs NP-Hardness and Inapproximability}\label{sec:complexity_different_actions}

Let us start our analysis with a negative result, showing that the problem of computing an O/PSE in SSCGs with different action spaces is computationally intractable even if the leader can only select a single resource and all the costs are monotonic functions of the resource congestion. This also shows that computing an equilibrium in an SCG which either maximizes or minimizes the usage of a resource (or the cost incurred by a player) is hard, which may be of independent interest.

First, we prove that finding an OSE is not in \PolyAPX~unless \Poly $=$ \NP, using a reduction from 3SAT.
Then, we show that the same intractability result holds for computing a PSE in SSCGs by means of a different reduction still based on 3SAT.

\subsection{Computational Complexity of Finding an OSE in SSCGs}\label{sub_sec:reduction_different_actions_opt}

First, we analyze the problem of computing an OSE in SSCGs in the general case with different action spaces.
The hardness and inapproximability results that we present are based on a reduction from 3SAT (see~\cite{garey1979computers} for its \NP-completeness), which is defined as follows:
\begin{definition}[3SAT]
	Given a finite set $C$ of 3-literal clauses defined over a finite set $V$ of Boolean variables, is there a truth assignment to the variables which satisfies all the clauses?
\end{definition}

In the following, let $l \in \phi$ denote a literal (i.e., a variable or its negation) appearing in clause $\phi \in C$ and $v(l) \in V$ denote the variable corresponding to that literal.
Moreover, given a 3SAT instance $(C,V)$, let $m$ and $s$ be, respectively, the number of clauses and variables, i.e., $m := |C|$ and $s := |V|$.

We introduce our reduction in the proof of the following theorem.
\begin{theorem}\label{thm:np_hard_opt}
	Computing an OSE in SSCGs with different action spaces is \NPHard.
\end{theorem}
\begin{proof}
	We provide a reduction from 3SAT showing that the existence of a polynomial-time algorithm for computing an OSE in SSCGs would allow us to solve any 3SAT instance in polynomial time.
	Specifically, given a 3SAT instance $(C,V)$ and a real number $0 < \epsilon < 4$, we build an instance $\Gamma_\epsilon(C,V)$ of an SSCG admitting an OSE in which the leader's cost is $\epsilon$ if and only if $(C,V)$ is satisfiable; if not, the leader's cost is $4$ in any OSE.
	%
	%
	
	\textbf{Mapping.}
	$\Gamma_\epsilon(C,V)$ is defined as follows:
	\begin{itemize}
		\item $N=F\cup \{\ell\}$, with $F=\{ p_\phi,p_{\phi,t} \mid \phi \in C \} \cup \{ p_{v} \mid v \in V\} \cup \{p_{v,k} ,p_{\bar v ,k} \mid v \in V, k \in \{ 1, \ldots,m \} \} \cup \{ p_{\phi,v},p_{\phi,\bar v} \mid \phi \in C, v \in V \}$;
		\item $R = \{ r_t \} \cup \{ r_\phi \mid \phi \in C \} \cup \{ r_{v}, r_{v,t}, r_{\bar v}, r_{\bar v,t} \mid v \in V \} \cup \{ r_{\phi,v},r_{\phi,\bar v} \mid \phi \in C, v\in V \}$;
		\item $A_{p_\phi} = \{ r_\phi \} \cup \{ r_{\phi,l} \mid l \in \phi \}, A_{p_{\phi,t}} = \{ r_\phi, r_t \} \ \ \forall\ \phi \in C$;
		\item $A_{p_{v, k}} = \{ r_{v,t},r_v \}, A_{p_{\bar v, k}} = \{ r_{\bar v,t}, r_{\bar v} \} \ \ \forall\ v \in V, k \in \{1,\ldots,m\}$;
		\item $A_{p_{v}} = \{r_t, r_{v,t}, r_{\bar v,t} \} \ \ \forall\ v \in V$;
		\item $A_{p_{\phi,v}} = \{ r_v,r_{\phi,v} \}, A_{p_{\phi,\bar v}} = \{ r_{\bar v}, r_{\phi,\bar v} \} \ \ \forall\ \phi \in C, v \in V$;
		\item $A_\ell = \{ r_t \}$.
	\end{itemize}
	The cost functions take values according to the following table, and satisfy $c_{r_{\bar v},f} = c_{r_{v},f}$, $c_{r_{\phi,\bar v},f} = c_{r_{\phi,v},f}$, $c_{r_{\bar v,t},f} = c_{r_{v,t},f}$, and $c_{r_{t},f} = c_{r_{t},\ell}$ (let us remark that they are all monotonic functions of the resource congestion):
	
        \bigskip
	%
	%

 	\begin{center}
	    {\renewcommand{\arraystretch}{1}\begin{tabular}{c|ccccc}
		\hline
		$x$ & $c_{r_\phi,f}$ & $c_{r_{v},f}$  & $c_{r_{v,t},f}$ & $c_{r_{\phi,v},f}$ &  $c_{r_{t},f}$ \\
		\cline{1-6}
		$1$ & $2$ & $0$ & $0$ & $1$ & $\epsilon$ \\
		$ [2,m] $ & $5$ & $0$ & $6$ & $6$ & $4$ \\
		$ [m+1,\infty] $ & $5$ & $7$ & $6$ & $6$  & $4$ \\
	    \end{tabular}}
        \end{center}
	
        \bigskip

	\noindent Figure~\ref{fig:reduction_opt} shows an example of the game instance $\Gamma_\epsilon(C,V)$.
	
	Given a 3SAT instance $(C,V)$, $\Gamma_\epsilon(C,V)$ can be constructed in polynomial time, as it features $n = 2m + s + 4 m s+ 1$ players and $r = m + 4 s + 2 m s + 1$ resources.
	Since, in $\Gamma_\epsilon(C,V)$, the leader can only select a single resource, $r_t$, the only leader's commitment is $\sigma_\ell(r_t) = 1$.
	As a result, the leader's cost is $\epsilon$ if and only if no follower selects resource $r_t$; otherwise, it is $4$.
	
	\begin{figure*}[!htp]
		\centering
		\includegraphics[width=\textwidth]{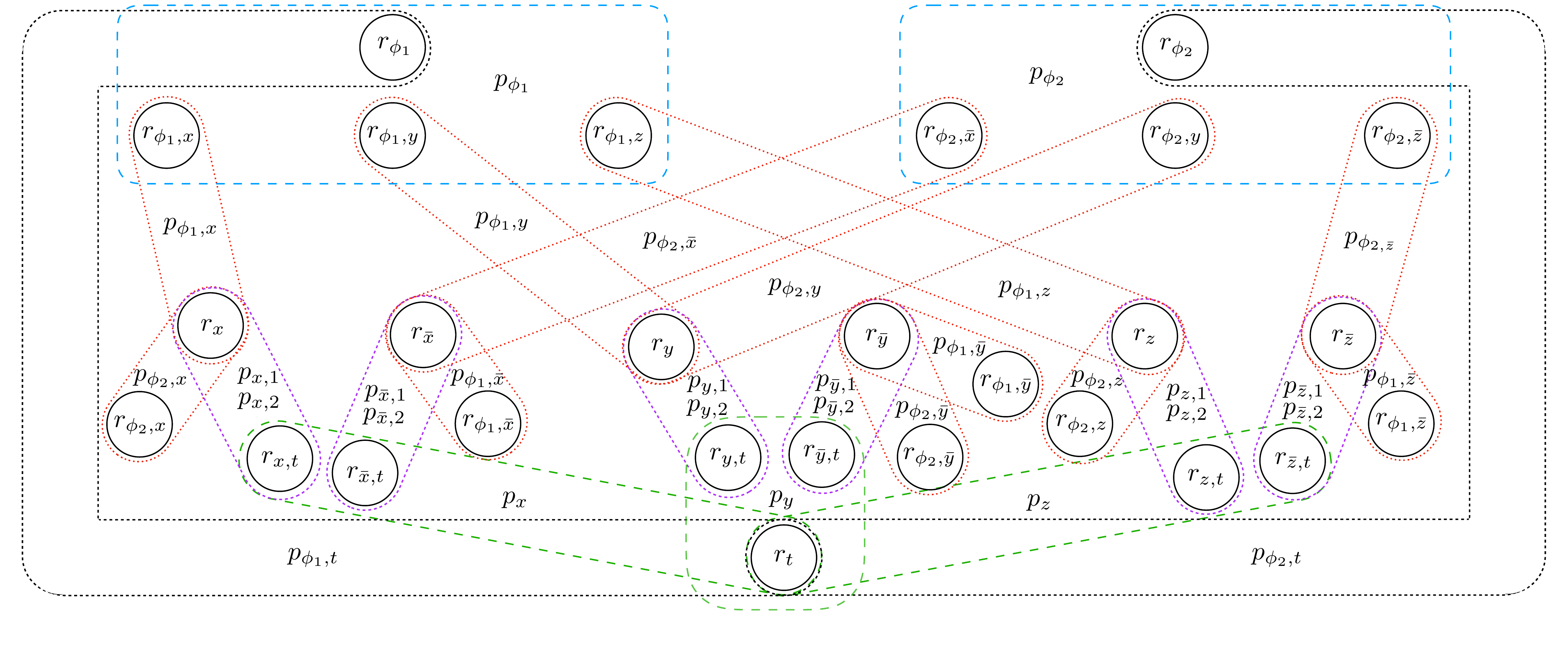}
		\caption{Example of a game instance $\Gamma_\epsilon(C,V)$ used in the reduction in the proof of Theorem~\ref{thm:np_hard_opt}, with $V=\{x,y,z\}$, $C = \{\phi_1,\phi_2\}$, $\phi_1 = x \vee y \vee z$, and $\phi_2 = \bar x \vee y \vee \bar z$.}
		\label{fig:reduction_opt}
	\end{figure*}

	\textbf{If.} 
	Assume that $(C,V)$ is satisfiable, and let $\tau : V \rightarrow \{\mathsf T,\mathsf F\}$ be a truth assignment satisfying all the clauses in $C$. 
	Using $\tau$, we show how to recover a followers' action profile $a = (a_p)_{p \in F} \in \bigtimes_{p \in F} A_p$ such that $a \in E^{\sigma_\ell}$, with $\sigma = (\sigma_\ell,a)$ providing the leader with a cost of $\epsilon$.
	Note that, since $\epsilon$ is the minimum cost the leader can achieve and the followers behave optimistically, $\sigma$ is an OSE.
	In particular, let $a_{p_{\phi,t}} = r_\phi$, for all $\phi \in C$.
	Moreover, if $\tau(v)=\mathsf T$, let $a_{p_v} = r_{\bar v,t}$ and $a_{p_{\phi,v}} = r_{v}$, $a_{p_{\phi,\bar v}} = r_{\phi, \bar v}$ for all $\phi \in C$, while, for all $k \in \{1,\ldots,m\}$, let $a_{p_{v,k}} = r_{v,t}$ and $a_{p_{\bar v,k}} = r_{\bar v}$.
	Instead, if $\tau(v)=\mathsf F$, let $a_{p_{\bar v}} = r_{v,t}$ and $a_{p_{\phi,\bar v}} = r_{\bar v}$, $a_{p_{\phi,v}} = r_{\phi,v}$ for all $\phi \in C$, while, for all $k \in \{1,\ldots,m\}$, let $a_{p_{\bar v,k}} = r_{\bar v,t}$ and $a_{p_{v,k}} = r_{v}$.
	Notice that, since either $\tau(v) = \mathsf T$ or $\tau(v) = \mathsf F$, two cases are possible. 
	If $\tau(v) = \mathsf T$, we have $\nu_{r_v}^a = m$ (followers $p_{\phi,v}$), $\nu_{r_{\bar v}}^a = m$ (followers $p_{\bar v,k}$), $\nu_{r_{v,t}}^a = m$ (followers $p_{v,k}$), and $\nu_{r_{\bar v,t}}^a = 1$ (follower $p_{v}$).
	If $\tau(v) = \mathsf F$, we have $\nu_{r_{\bar v}}^a = m$ (followers $p_{\phi,\bar v}$), $\nu_{r_{v}}^a = m$ (followers $p_{v,k}$), $\nu_{r_{\bar v,t}}^a = m$ (followers $p_{\bar v,k}$), and $\nu_{r_{v,t}}^a = 1$ (follower $p_{v}$).
	%
	%
	Assume, w.l.o.g., $\tau(v) = \mathsf T$, as the other case is analogous.
	%
	%
	First, no follower $p_{ \phi,v }$ would deviate from $r_v$ to $r_{\phi,v}$, as, otherwise, she would incur a cost of at least $1$, rather than $0$. The same holds for followers $p_{\phi,\bar v}$, as their cost is at most $6$ while, if any of them switched to $r_{\bar v}$, she would incur a cost of $7$.
	Similarly, followers $p_{v,k}$ would not deviate from $r_{v,t}$ (as $6  < 7$) and followers $p_{\bar v,k}$ would not deviate from $r_{\bar v}$  (as $0 < 6$).
	Since $\nu_{r_{\bar v,t}}^a = 1$, follower $p_{v}$ would not deviate from $r_{\bar v,t}, $ (as $0 < 6$ and $0 < 4$ ).
	Furthermore, since $\tau$ is a truth assignment satisfying $(C,V)$, at least one literal $l \in \phi$ evaluates to true under $\tau$ for every $\phi \in C$.
	Let $a_{p_\phi} = r_{\phi,l}$ for every $\phi \in C$.
	Since $l$ evaluates to true, it must be $a_{p_{\phi,l}} = r_{l}$, thus $p_\phi$ is the only follower who selects $r_{\phi,l}$.
	As a result, $p_\phi$ incurs a cost of $1$, and she has no incentive to deviate.
	Finally, $p_{\phi,t}$ does not deviate from $r_{\phi}$ to $r_t$ as $2 < 4$.
	Thus, we can conclude that $a$ is an NE and that, since no follower chose $r_t$, the leader's cost is $\epsilon$.

	\textbf{Only if.}
	Suppose there exists an OSE $\sigma = (\sigma_\ell,a)$ in which the leader's cost is $\epsilon$.
	We show that, in polynomial time, one can recover a truth assignment $\tau$ that satisfies all the clauses in $C$ from $a = (a_p)_{p \in F} \in \bigtimes_{p \in F} A_p$.
	First, let us note that no follower selects $r_t$ in $a$ as, otherwise, the leader's cost would be $4 > \epsilon$.
	As a consequence, all followers $p_{\phi,t}$ and $p_{v}$ must select one of the other resources available to them, i.e, $a_{p_{\phi,t}} = r_\phi$ and $a_{p_{v}} \in \{r_{v,t}$,$r_{\bar v,t}\}$.
	Moreover, there cannot be two followers using resource $r_\phi$ for every $\phi \in C$ as, otherwise, $p_{\phi,t}$ would have an incentive to deviate from $r_\phi$ to $r_t$ (as $5 > 4$).
	Thus, $a_{p_\phi} \neq r_\phi$, and, for all $\phi \in C$, there must be a literal $l \in \phi$ such that $a_{p_\phi} = r_{\phi,l}$.
	In addition, there cannot be two followers selecting $r_{\phi,l}$ as, otherwise, $p_\phi$ would have an incentive to deviate to $r_\phi$ (as $5 < 6$).
	Thus, it must be the case that $a_{p_{\phi, l}} = r_l$.
	This implies that $\nu_{r_l}^a \le m$ as, otherwise, the cost of $p_{\phi, l}$ would be $7 > 6$, and that follower would change resource, switching to $r_{\phi, l}$.
	Thus, at least one of the followers $p_{l,k}$ must select  $r_{l,t}$ as, otherwise, $\nu_{r_l}^a > m$.
	As a consequence, if $l$ is positive and $v(l) = v$, $p_v$ selects $r_{\bar v,t}$ as, if she selected $r_{v,t}$, she would have an incentive to deviate (as $6 > 4$).
	Moreover, no other follower would select $r_{\bar v,t}$ as, otherwise, $p_{v}$ would deviate to $r_t$ (as $6 > 4$).
	This implies that $\nu_{r_{\bar v,t}}^a =1$ (follower $p_{v}$) and all the followers $p_{\bar v,k} $ select resource $r_{\bar v}$, while the followers $p_{\phi,\bar v}$ choose resources $r_{\phi,\bar v}$.
	On the other hand, if $l$ is negative and $v(l) = v$, similar arguments allow us to conclude that $\nu_{r_{v,t}}^a =1$ (follower $p_{v}$) and all the followers $p_{v,k} $ select resource $r_{v}$, while the followers $p_{\phi,v}$ choose resources $r_{\phi,v}$.
	As a result, either $\nu_{r_{v,t}}^a =1$ or $\nu_{r_{\bar v,t}}^a =1$.
	%
	%
	In conclusion, we can define a truth assignment $\tau$ such that $\tau(v)=\mathsf T$ if $a_{p_v}=r_{\bar v,t}$ and $\tau(v)= \mathsf F$ if $a_{p_v}=r_{v,t}$.
	Clearly, $\tau$ is well-defined. 
	Moreover, as previously shown, for every $\phi \in C$ there exists a literal $l \in \phi$ such that $a_{p_{\phi, l}} = r_{l}$, which, letting $v=v(l)$, implies that $\nu_{r_{\bar v,t}}^a =1$. Thus, $\tau(v(l))=\mathsf T$ if $l$ is positive, while $\nu_{r_{ v,t}}^a =1$ and $\tau(v(l))=\mathsf F$ if $l$ is negative.
	Hence, $\tau$ satisfies all the clauses.
\end{proof}

The proof of Theorem~\ref{thm:np_hard_opt} also shows the following: 
\begin{corollary}\label{cor:ne_opt_np_hard}
	In general SCGs without leadership and different action spaces, computing an NE minimizing the cost of a given player is $\mathsf{NP}$-hard.
\end{corollary}
\begin{proof}
	The result is easily proved by noticing that, in the $\Gamma_\epsilon(C,V)$ games defined in the proof of Theorem~\ref{thm:np_hard_opt}, since the leader can only use a single resource any OSE is also an NE.
	Thus, given that the followers behave optimistically, such games admit an optimal NE with leader's cost $\epsilon$ if and only if the corresponding 3SAT instance is satisfiable.
\end{proof}


Theorem~\ref{thm:np_hard_opt} also implies that the leader's cost in an OSE cannot be efficiently approximated up to within any factor which depends polynomially on the size of the input:

\begin{corollary}\label{cor:non_poly_apx_opt}
	The problem of computing an  OSE in SSCGs with different action spaces is not in Poly-$\mathsf{APX}$ unless $\mathsf{P} = \mathsf{NP}$.
\end{corollary}

\begin{proof}
	Given a 3SAT instance $(C,V)$, let us build an SSCG $\Gamma_\epsilon(C,V)$ as in the proof of Theorem~\ref{thm:np_hard_opt}.
	We have already proven that $\Gamma_\epsilon(C,V)$ admits an OSE in which the leader's cost is $\epsilon$ if and only if $(C,V)$ is satisfiable and that, otherwise, the leader's cost is $4$.
	Let $\epsilon = \frac{4}{2^{n+r}}$.
	Assume that there exists a polynomial-time approximation algorithm $\mathcal{A}$ with approximation factor $\text{poly}(n,r)$, i.e., a polynomial function of $n$ and $r$.
	%
	%
	%
	%
	Assume $(C,V)$ is satisfiable. $\mathcal{A}$ applied to $\Gamma_\epsilon(C,V)$ would return a solution with leader's cost at most $\frac{4}{2^{n+r}} \ \text{poly}(n,r)$. Since, for $n$ and $r$ large enough, $\frac{4}{2^{n+r}} \ \text{poly}(n,r) < 4$, $\mathcal{A}$ would allows us to decide in polynomial time whether $(C,V)$ is satisfiable, a contradiction unless $\mathsf{P} = \mathsf{NP}$.
\end{proof}

Since the followers break ties in favor of the leader in the reduction, the results in Theorem~\ref{thm:np_hard_opt} and Corollary~\ref{cor:non_poly_apx_opt} do not apply to the problem of finding a PSE.
We consider this case in the next subsection.

\subsection{Computational Complexity of Finding a PSE in SSCGs}\label{sub_sec:reduction_different_actions_pes}

The hardness and inapproximability results that we are about to present for the problem of computing a PSE in SSSCGs are still based on 3SAT, but rely on a different reduction.

%



%
\begin{theorem}\label{thm:np_hard_pes}
  Computing a PSE in general SSCGs with different action spaces is \NPHard.
	\end{theorem}
\begin{proof}
	%
	We provide a reduction from
3SAT
showing that the existence of a polynomial-time algorithm for computing a PSE in SSCGs would allow us to solve any
3SAT instance in polynomial time.
	Specifically, given a
        3SAT
        instance $(C,V)$ and a real number $0 < \epsilon < 4$, we build an SSCG instance $\Gamma_\epsilon(C,V)$ such that it admits a PSE where the leader's cost is $\epsilon$ if and only if the
        3SAT instance admits a \emph{no}
        answer, i.e., if and only if $(C,V)$ is not satisfiable.
	Instead, if the
        3SAT instance has answer \emph{yes},
        i.e., if $(C,V)$ is satisfiable, then the leader's cost is $4$ in any PSE.
	%
	%
	
	\textbf{Mapping.}
	$\Gamma_\epsilon(C,V)$ is defined as follows:
	\begin{itemize}
		\item $N=F\cup \{\ell\}$, with $F=\{ p_{l,\phi} \mid \phi \in C, l \in \phi \} \cup \{p_{\phi,t} \mid \phi \in C\}  \cup \{ p_{v,t}, p_{v}, p_{\bar v} \mid v \in V \} \cup \{ p_{\phi,v},p_{\phi,\bar v} \mid \phi \in C, v \in V \}$;
		\item $R = \{ r_t \} \cup \{ r_\phi \mid \phi \in C \} \cup \{ r_{v,t},r_{v},r_{\bar v} \mid v \in V \} \cup \{ r_{\phi,v},r_{\phi,\bar v} \mid \phi \in C, v\in V \}$;
		\item $A_{p_{l, \phi}} = \{ r_\phi \} \cup \{ r_{\phi,l}\} \ \ \forall\ \phi \in C, \forall\ l \in \phi$;
		\item $A_{p_{\phi,t}} = \{ r_\phi, r_t \} \ \ \forall\ \phi \in C$;
		\item $A_{p_v} = \{ r_{v,t},r_v \}, A_{p_{\bar v}} = \{ r_{v,t}, r_{\bar v} \}, A_{p_{v,t}} = \{ r_{v,t},r_t \} \ \ \forall\ v \in V$;
		\item $A_{p_{\phi,v}} = \{ r_v,r_{\phi,v} \}, A_{p_{\phi,\bar v}} = \{ r_{\bar v}, r_{\phi,\bar v} \} \ \ \forall\ \phi \in C, v \in V$;
		\item $A_\ell = \{ r_t \}$.
	\end{itemize}
	The cost functions take values according to the following table, and satisfy $c_{r_{\bar v},f} = c_{r_{v},f}$, $c_{r_{\phi,\bar v},f} = c_{r_{\phi,v},f}$, and $c_{r_{t},f} = c_{r_{t},\ell}$ (let us remark that they are all monotonic functions of the resource congestion):

        \bigskip
	%
	%
	%
	%

	\begin{center}
	  {\renewcommand{\arraystretch}{1}\begin{tabular}{c|ccccc}
	      \hline
	      $x$ & $c_{r_\phi,f}$ & $c_{r_{v},f}$  & $c_{r_{v,t},f}$ & $c_{r_{\phi,v},f}$ &  $c_{r_{t},f}$ \\
	      \cline{1-6}
	      $1$ & $2$ & $1$ & $2$ & $0$ & $\epsilon$ \\
	      $ [2,m] $ & $5$ & $1$ & $5$ & $7$ & $\epsilon$ \\
	      $ m+1 $ & $5$ & $6$ & $5$ & $7$  & $\epsilon$ \\
	      $ [m+s+1, \infty] $ & $5$ & $6$ & $5$ & $7$  & $4$ \\
	  \end{tabular}}
	\end{center}
       
        \bigskip
	
	\noindent Figure~\ref{fig:reduction_pes} shows an example of game $\Gamma_\epsilon(C,V)$.
	
	Given $(C,V)$,  $\Gamma_\epsilon(C,V)$ can be constructed in polynomial time, as it features $n = 3m + m + 3s + 2 m s + 1$ players and $r = m + 3 s + 2 m s + 1$ resources.
	
	Observe that, in $\Gamma_\epsilon(C,V)$, the leader can only select a single resource $r_t$ and, hence, the only leader's commitment is $\sigma_\ell(r_t) = 1$.
	As a result, the leader's cost is $4$ if and only if all followers $p_{\pi, t}$ and $p_{v,t}$ select resource $r_t$; otherwise, it is $\epsilon$. 

	\begin{figure*}[!htp]
		\centering
		\includegraphics[width=\textwidth]{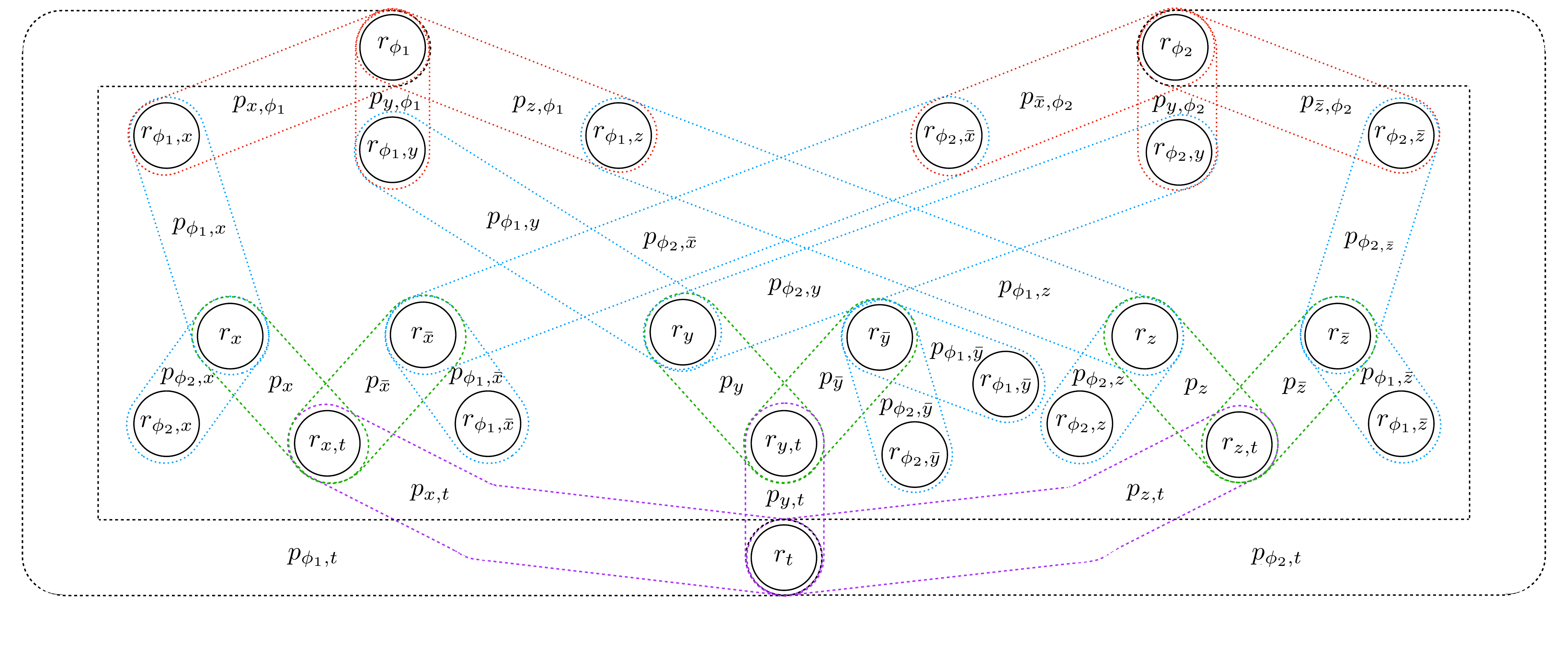}
		\caption{Example of a game instance $\Gamma_\epsilon(C,V)$ used in the reduction in the proof of Theorem~\ref{thm:np_hard_pes} with $V=\{x,y,z\}$, $C = \{\phi_1,\phi_2\}$, $\phi_1 = x \vee y \vee z$, and $\phi_2 = \bar x \vee y \vee \bar z$.}
		\label{fig:reduction_pes}
	\end{figure*}

	\textbf{If.}
	Suppose that the
        3SAT instance has answer \emph{no},
        i.e., there is no truth assignment to the variables in $V$ that satisfies all the clauses in $C$.
	We prove that, in that case, $\Gamma_\epsilon(C,V)$ admits a PSE with leader's cost equal to $\epsilon$.
	By contradiction, let us assume there exists a PSE $\sigma = (\sigma_\ell,a)$ in which the leader's cost is $4 > \epsilon$.
	We show that $a = (a_p)_{p \in F} \in \bigtimes_{p \in F} A_p$ can be employed to recover, in polynomial time, a truth assignment $\tau$ that satisfies all the clauses in $C$, which is a contradiction.
	First, let us note that all the followers $p_{\pi, t}$ and $p_{v,t}$ select $r_t$ in $a$ as, otherwise, the leader's cost would be $\epsilon < 4$.
	As a result, $a_{p_\phi} = r_t$ for every $\phi \in C$, and $a_{p_v} = r_t$ for all $v \in V$.
	Thus, there is at least one follower using resource $r_\phi$ for every $\phi \in C$ as, otherwise, $p_{\phi,t}$ would have an incentive to deviate from $r_t$ to $r_\phi$ (as $4 > 2$).
	Thus, for each $\phi \in C$ there must be a literal $l \in \phi$ such that $a_{p_{l,\phi}} = r_{\phi}$.
	This implies that $a_{p_\phi,l} = r_{\phi,l}$ as, otherwise, follower $p_{l,\phi}$ would deviate from  $r_{\phi}$ to $r_{\phi,l}$ (as $2 > 0$).
	Thus, we can conclude that $\nu_{r_l}^a < m$.
	As a consequence, $a_{p_l} = r_{l}$ as, if she selected $r_{v(l),t}$, her cost would be $2$ or more, and she would deviate to resource $r_{l}$ (to incur a smaller cost equal to 1).
	Furthermore, at least one between $p_v$ and $p_{\bar v}$ must select $r_{v,t}$ as, otherwise, player $p_{v,t}$ would deviate from $r_t$ (as $2 < 4$).
	Assume, w.l.o.g., $a_{p_{\bar v}} = r_{v,t}$, as the other case is analogous. 
	All the followers $p_{\phi, \bar v}$ select $r_{\bar v}$ as, otherwise, $p_{\bar v}$ would have an incentive to deviate from $r_{v,t}$ (as $1 < 2$).
	Thus, all the followers $p_{\bar v,\phi}$ select $r_{\phi,\bar v}$.
	Let us define a truth assignment $\tau$ such that $\tau(v)=\mathsf T$ if $a_{p_v}=r_v$, $\tau(v)= \mathsf F$ if $a_{p_{\bar v}}=r_{\bar v}$ and $\tau(v)$ is either $\mathsf T$ or $\mathsf F$ whenever $a_{p_v}=a_{p_{\bar v}}= r_{v,t}$.
	Clearly, $\tau$ is well-defined. 
	Moreover, as previously shown, for every $\phi \in C$ there exists a literal $l \in \phi$ such that $a_{p_{l,\phi}} = r_{\phi}$. This implies $a_{p_{\phi,l}} = r_{\phi,l}$ and $a_{p_l}=r_l$, and, thus, $\tau(v(l))=\mathsf T$ if $l$ is positive or $\tau(v(l))=\mathsf F$ if it is negative.
	Therefore, $\tau$ satisfies all the clauses, which is a contradiction.
	
	\textbf{Only if.} 
	Suppose that the
        3SAT instance admits answer \emph{yes},
        i.e., there exists a truth assignment to the variables which satisfies all the clauses in $C$.
	We prove that in any PSE of $\Gamma_\epsilon(C,V)$ the leader's cost is $4 > \epsilon$.
	Let $\tau : V \rightarrow \{\mathsf T,\mathsf F\}$ be one such truth assignment. 
	We show how to recover from $\tau$ a followers' action profile $a = (a_p)_{p \in F} \in \bigtimes_{p \in F} A_p$ such that $a \in E^{\sigma_\ell}$, with $\sigma = (\sigma_\ell,a)$ providing the leader with a cost of $4$.
	Since $4$ is the maximum cost the leader can achieve and the followers behave pessimistically, $\sigma$ is clearly a PSE.
	In particular, let $a_{p_{\phi,t}} = r_t$, for all $\phi \in C$, and $a_{p_{v,t}} =r_t$, for all $v \in V$.
	Moreover, if $\tau(v)=\mathsf T$, let $a_{p_v} = r_v$, $a_{p_{\bar v}} = r_{v,t}$, and, for all $\phi \in C$, $a_{p_{\phi,v}} = r_{\phi,v} $ and $a_{p_{\phi,\bar v}} = r_{\bar v}$.
	Additionally, for every clause $\phi \in C$ and $l \in \phi$ such that $v(l) = v$, let $a_{p_{l, \phi}} = r_{\phi}$ if $l$ is positive, while $a_{p_{l, \phi}} = r_{\phi,l}$ if it is negative.
	Conversely, if $\tau(v)=\mathsf F$, let $a_{p_{\bar v}} = r_{\bar v}$, $a_{p_v} = r_{v,t}$, and, for all $\phi \in C$, $a_{p_{\phi,\bar v}} = r_{\phi,\bar v}$ and $a_{p_{\phi,v}} = r_{v}$. 
	Furthermore, for every clause $\phi \in C$ and $l \in \phi$ such that $v(l) = v$, let $a_{p_{l, \phi}} = r_{\phi}$ if $l$ is negative, and $a_{p_{l, \phi}} = r_{\phi,l}$ if it is positive.
	Notice that, since either $\tau(v) = \mathsf T$ or $\tau(v) = \mathsf F$, one between $p_v$ and $p_{\bar v}$ selects $r_{v,t}$.
	Assume, w.l.o.g., $a_{p_{\bar v}} = r_{v,t}$ and $a_{p_v} = r_v$ (as the other case is analogous).
	First, no follower $p_{ \phi,v }$ would deviate from $r_{\phi,v} $ to $r_v$, as, otherwise, she would incur a cost of at least $1$, rather than $0$. 
	The same holds for followers $p_{\phi,\bar v}$, as their cost is $1$ while, if any of them switched to $r_{\phi,\bar v}$, she would incur a cost of $7$, because $a_{p_{\bar v,\phi}} = r_{\phi,\bar v}$.
	Similarly, since there is one follower selecting $r_{v,t}$, follower $p_{v,t}$ would not deviate from $r_t$ (as $4 < 5$), while follower $p_{v}$ would not deviate from $r_{v}$ because her cost is $1 < 5$ and $p_{\bar v}$ would not switch from $r_{v,t}$ (as she would get $6$ rather than $1$).
	Furthermore, since $\tau$ is a truth assignment satisfying $(C,V)$, for each clause $\phi \in C$ there exists at least one literal $l \in \phi$ that evaluates to true under $\tau$. Thus, $p_{l, \phi}$ would not deviate from $r_{\phi}$ (as she pays either $2$ or $5$ instead of $7$).
	Thus all the followers $p_{\phi,t}$ would not deviate from $r_{t}$ (as $ 4 < 5$) and we can conclude that $a$ is an NE. Since $m + s$ follower use $r_t$, the leader's cost is $4$.
\end{proof}
		

Theorem~\ref{thm:np_hard_pes} also implies the following: 
\begin{corollary}\label{cor:ne_pes_np_hard}
	In SCGs without leadership and different action spaces, computing an NE maximizing the cost of a given player is $\mathsf{NP}$-hard.
\end{corollary}
\begin{proof}
	In games $\Gamma_\epsilon(C,V)$ such as those used in the proof of Theorem~\ref{thm:np_hard_pes}, any PSE is also an NE (since the leader can choose a single action).
	Moreover, $\Gamma_\epsilon(C,V)$ admits a PSE where the leader's cost is 4 if and only if the given
        3SAT instance has answer \emph{no},
        which proves the result.
\end{proof}


Furthermore, from Theorem~\ref{thm:np_hard_pes} it directly follows that the leader's cost in a PSE cannot be efficiently approximated up to within any approximation factor which depends polynomially on the size of the input:
\begin{corollary}
	The problem of computing a PSE in SSCGs with different action spaces is not in Poly-$\mathsf{APX}$ unless $\mathsf{P} = \mathsf{NP}$.
\end{corollary}

\begin{proof}
	Given a
        3SAT
        instance $(C,V)$, let us build an instance $\Gamma_\epsilon(C,V)$ of an SSCG as in the proof of Theorem~\ref{thm:np_hard_pes}.
	We have already proven that $\Gamma_\epsilon(C,V)$ admits a PSE in which the leader's cost is $\epsilon$ if and only if the
        3SAT instance has answer \emph{no}; 
        otherwise, the leader's cost is $4$ in any PSE.
	Let $\epsilon = \frac{4}{2^{n+r}}$.
	Assume that there exists a polynomial-time approximation algorithm $\mathcal{A}$ with approximation factor $\text{poly}(n,r)$, i.e., a polynomial function of $n$ and $r$.
	Assume the answer to the
        3SAT instance is \emph{no}.
        $\mathcal{A}$ applied to $\Gamma_\epsilon(C,V)$ would return a solution with leader's cost at most $\frac{4}{2^{n+r}} \ \text{poly}(n,r)$. Since, for $n$ and $r$ large enough, $\frac{4}{2^{n+r}} \ \text{poly}(n,r) < 4$, $\mathcal{A}$ would allows us to decide in polynomial time whether the answer to the
        3SAT instance is \emph{yes} or \emph{no}, a contradiction unless $\mathsf{P} = \mathsf{NP}$.
\end{proof}

\section{SSSCGs NP-Hardness and Inapproximability}\label{sec:complexity_same_actions}

We focus, in this section, on SSSCGs (the subset of SSCGs in which the players have the same action spaces), showing that the problem of finding an O/PSE in such games is \NPHard~and not in \PolyAPX~unless \Poly $=$ \NP.
For the problem of computing an OSE, we rely on $K$--PARTITION, a variant of PARTITION with an additional size constraint, whereas we adopt the classical version of PARTITION for the problem of computing a PSE.
The two problems are defined as follows:
\begin{definition}[PARTITION]\label{def:partition}
	Given a finite set $S = \{ x_1, \ldots, x_{|S|} \}$ of positive integers $x_i \in \mathbb{Z}^+$ with both $|S|$ and $\sum_{i \in S} x_i$ even, is there a partition $(S',S \setminus S')$ of $S$, with $S' \subseteq S$, such that $ \sum_{x_i \in S'} x_i = \sum_{x_i \in S \setminus S'} x_i$?
\end{definition}
\begin{definition}[$K$--PARTITION]
	Given a finite set $S = \{ x_1, \ldots, x_{|S|} \}$ of positive integers $x_i \in \mathbb{Z}^+$ with both $|S|$ and $\sum_{i \in S} x_i$ even and a positive integer
$K \leq \frac{|S|}{2}$,
is there a partition $(S',S \setminus S')$ of $S$, with $S' \subseteq S$ and $|S'|=K$, such that $ \sum_{x_i \in S'} x_i = \sum_{x_i \in S \setminus S'} x_i$?
	%
\end{definition}

Letting $s=\frac{1}{2}\sum_{x_i \in S} x_i $, we assume for both problems that $x_i \leq s$ for all $i \in S$. Indeed, if some $x_i > s$, then $\sum_{x_i \in S'} x_i > s$ for every $S' \subseteq S$ and, thus, the answer to both PARTITION and $K$--PARTITION is trivially {\em no}.


PARTITION is well-known to be \NP-complete~\cite{garey1979computers}.
To see that $K$--PARTITION is also \NP-complete (its membership to \NP~is clear), it suffices to observe that PARTITION has answer {\em yes} if and only if $K$--PARTITION has answer {\em yes} for some $K \in \left\{1, \dots, \frac{|S|}{2}\right\}$.
This gives us a simple Cook reduction from PARTITION to $K$--PARTITION: after solving $K$--PARTITION $\frac{|S|}{2}$ times, once per value of $K \in \left\{1, \dots, \frac{|S|}{2}\right\}$, if answer {\em yes} is found for some $K$, PARTITION has answer {\em yes}; if, conversely, answer {\em yes} is never found, PARTITION has answer {\em no}.

\subsection{Computational Complexity of Finding an OSE in SSSCGs}\label{sub_sec:reduction_same_actions_opt}

We start our analysis with the problem of computing an OSE in SSSCGs.
%
%
%
We introduce our main reduction in the proof of the following theorem.

\begin{theorem}\label{thm:np_hard_opt_same}
	Computing an OSE in SSSCGs is \NPHard.
\end{theorem}

\begin{proof}
	We prove the theorem using a reduction from $K$--PARTITION, showing that the existence of a polynomial-time algorithm for computing an OSE in SSSCGs would allow us to solve $K$--PARTITION in polynomial time.
	Let us recall that $s$ is defined as $s= \frac{1}{2} \sum_{x_i \in S} x_i$.
	Clearly, any solution $(S,S')$ to K--PARTITION is uniquely defined by a subset $S' \subseteq S$ such that $\sum_{x_i \in S'} x_i = s$ and $|S'| = K$.
	Let $w_i=\frac{x_i}{s}$ for all $x_i \in S$. Due to having $x_i \leq s$ for all $x_i \in S$, we also have $w_i \leq 1$.
	Given a $K$--PARTITION instance defined by a pair $(S,K)$, we build an instance $\Gamma_\epsilon(S,K)$ of an SSSCG with $0 < \epsilon<1$ such that the leader's cost in an OSE is $\epsilon$ if and only if $(S,K)$ admits answer \emph{yes}.
	
	\textbf{Mapping.} 	
	$\Gamma_\epsilon(S,K)$ is defined as follows:
	\begin{itemize}
		\item $N = F \cup \{\ell \}$, with $|F| = 4|S|+2$;
		\item $R = \{ r_{t_1}\} \cup \{r_{t_2}\} \cup \{ r_i \mid x_i \in S \}$;
	\end{itemize}
	The players' cost functions are specified in the following table:
	
	\begin{center}
	  {\renewcommand{\arraystretch}{1}\begin{tabular}{c|cc|cc|cc}
	      \hline
	      $x$ & $c_{r_i,f}$ & $c_{r_i,\ell}$  & $c_{r_{t_1},f}$ & $c_{r_t,\ell}$ & $c_{r_{t_2},f}$ & $c_{r_{t_2},\ell}$  \\
	      \cline{1-7}
	      $1$ & $2s$ & $s$ & $3s^2$ & $s^4$ & $1$ & $s^4$  \\
	      $ 2 $ & $0$ & $s$ & $3s^2$ & $s^4$  & $4s^2$ & $s^4$  \\
	      $ 3 $ & $\frac{1}{w_i}$ & $\epsilon$ & $3s^2$ & $s^4$ &$4s^2$ & $s^4$ \\
	      $ 4 $ & $\frac{2s-\frac{1}{w_i}+1}{w_i}$ & $s$ & $3s^2$ & $s^4$ & $4s^2$ & $s^4$  \\
	      $ [5,4 |S|-2K] $ & $4s^2$ & $s$ & $3s^2$ & $s^4$ & $4s^2$ & $s^4$\\
	      $  4 |S|-2K+1$ & $4s^2$ & $s$ & $2s$ & $s^4$ & $4s^2$ & $s^4$  \\
	      $ 4 |S|-2K+2 $ & $4s^2$ & $s$ & $1$ & $s^4$ & $4s^2$ & $s^4$\\
	      $ [4 |S|-2K+3,\infty] $ & $4s^2$ & $s$ & $0$ & $s^4$ & $4s^2$ & $s^4$\\
	  \end{tabular}}
	\end{center}
	
	\bigskip
		
	Clearly, $\Gamma_\epsilon(S,K)$ can be built in polynomial time, as it features $n=4|S|+3$ players and $r=|S|+2$ resources.
	
	\textbf{If.} 
	Suppose that the $K$--PARTITION instance $(S,K)$ admits a \emph{yes} answer.
	Let $S' \subseteq S$ be a set of integers with $|S'|=K$ and $\sum_{x_i \in S'} x_i = s$. 
	We prove that $\Gamma_\epsilon(S,K)$ admits an OSE in which the leader's cost is $\epsilon$.
	Given $S'$, let us define a followers' configuration $\nu \in \mathbb{R}^r$ and a leader's strategy $\sigma_\ell \in \Delta_\ell$ such that $\sigma = (\sigma_\ell, \nu)$ is an OSE with leader's cost $\epsilon$.
	Let $\nu_{r_i}=2$ and $\sigma_\ell(r_i)=w_i$ for all $x_i \in S'$, while, for every $x_i \notin S'$, let $\nu_{r_i}=0$ and $\sigma_\ell(r_i)=0$.
	Moreover, let $\nu_{r_{t_1}}=4 |S|-2K+1$, $\sigma_\ell(r_{t_1})=0$, $\nu_{r_{t_2}}=1$, and $\sigma_\ell(r_{t_2})=0$.
	First, let us observe that the leader's strategy $\sigma_\ell$ is well-defined, as 
	$$
	\sum_{x_i \in S} \sigma_\ell(r_i)+ \sigma_\ell(r_{t_1})+\sigma_\ell(r_{t_2})=\sum_{x_i \in S'} \sigma_\ell(r_i) = \sum_{x_i \in S'} w_i= \sum_{x_i \in S'} \frac{x_i}{s}=1,
	$$
	where the last equality follows from the fact that $S'$, with its complement $S \setminus S'$, defines a partition of $S$.
	Next, we show that $\nu$ is an NE for $\sigma_\ell$ with the following argument.
	\begin{itemize}
		\item All the followers who selected resource $r_i$, with $x_i \in S'$, do not have any incentive to change resource, as their cost is $w_i \cdot \frac{1}{w_i} = 1$ and they cannot improve it by switching to another resource.
		Indeed, if they selected a resource $r_{j}$ with $x_j \in S'$, they would incur a cost of $ \frac{1}{w_{j}} \cdot (1-w_{j})+\frac{2s-\frac{1}{w_{j}}+1}{w_{j}} \cdot w_{j}=2s > 1$.
		Similarly, their cost would be $2s$ if they choose $r_j$ with $x_j \notin S'$. 
		They would not benefit from choosing resource $r_{t_1}$, as they would incur a cost of $1$, which is the same as their current cost, and they would not switch to resource $r_{t_2}$, as their cost would become $4s^2 > 1$.
		\item All the followers who selected resource $r_{t_1}$ incur a cost of $2s$.
		Thus, they do not have an incentive to deviate to a resource $r_i$ with $x_i \in S'$, as they would still incur a cost of $2s$.
		The same holds for resource $r_{t_1}$.
		Similarly, if they chose to play $r_{t_2}$, they would incur a cost of $4s^2 > 2s$. 
		\item The follower who chose resource $r_{t_2}$ does not deviate, as her cost is $1$ and she would incur a cost of $2s$ and $1$ if she switched to resource $r_i$ or $r_{t_1}$, respectively.
	\end{itemize}
	Overall, the leader's cost is:
	\begin{align*}
			c_\ell^\sigma & = \sum_{x_i \in S} \sigma_\ell(r_i) c_{r_i,\ell}(\nu_{r_i}+1)+ \sigma_\ell(r_{t_1}) c_{r_{t_1},\ell}(\nu_{r_{t_1}}+1) + \sigma_\ell(r_{t_2}) c_{r_{t_2},\ell}(\nu_{r_{t_2}}+1) = \\
			& = \sum_{x_i \in S'} \sigma_\ell(r_i) c_{r_i,\ell}(\nu_{r_i}+1) = \sum_{x_i \in S'} \epsilon  w_i=\epsilon.
	\end{align*}

	\textbf{Only if.}
	Suppose that $\Gamma_\epsilon(S,K)$ has an OSE $\sigma = (\sigma_\ell,\nu)$ in which the leader's cost is $\epsilon$. 
	Then, $\sigma_\ell(r_{t_1})=\sigma_\ell(r_{t_2})=0$ must hold.
	Moreover, the leader must place positive probability only on resources $r_i$ with $\nu_{r_i}=2$.
	Clearly, there is always a resource $r_i$ with $\nu_{r_i}=2$ and $\sigma_\ell(r_i) > 0$.
	Next, we prove that $\nu_{r_{t_1}} = 4 |S|-2K+1$.
	By contradiction, assume that $\nu_{r_{t_1}} \neq 4 |S|-2K+1$. Three cases are possible.
	\begin{itemize}
		\item $\nu_{r_{t_1}} = 0$ implies that either there exists at least one resource $r_i$ with $\nu_{r_i} \ge 5$ or $\nu_{r_{t_2}}=2$, but, then, the followers who chose $r_i$ or, respectively, $r_{t_2}$, would deviate by choosing $r_{t_1}$, decreasing their cost from $4s^2$ to $3s^2$.
		\item $1 \leq \nu_{r_{t_1}} \le 4 |S|-2K$ implies that the followers who selected $r_{t_1}$ incur a cost of $3 s^2$. Thus, they would deviate to some resource $r_i$ with $\nu_{r_i}=2$, since their cost would be at most $\frac{2s-\frac{1}{w_i}+1}{w_i} < 3s^2$.
		\item $\nu_{r_{t_1}} \ge 4 |S|-2K+2$ implies that the followers' cost when they deviate by playing resource $r_t$ is $0$. Thus, the followers who selected a resource $r_i$ with $\nu_{r_i}=2$ and $\sigma_\ell(r_i)>0$ would change resource, since their current cost is strictly greater than $0$.
	\end{itemize}
	%
	%
	%
	%
	The only remaining option for $\nu$ to be an NE for $\sigma_\ell$ is $\nu_{r_{t_1}} = 4 |S|-2K+1$.
	Then, $\nu_{r_{t_2}} = 1$ must hold as, if $\nu_{r_{t_2}} = 0$, a follower would switch form resource $r_{t_1}$ to resource $r_{t_2}$ (incurring a cost of $1$ instead of one of $2s > 1$), while, if $\nu_{r_{t_2}} \geq 2$, the followers who selected resource $r_{t_2}$ would deviate to resource $r_{t_1}$ (incurring a cost of $1$ instead of one of $4s^2 > 1$).
	%
	%
	%
	Let us now consider a resource $r_i$ with $\nu_{r_i}=2$. 
	We prove that $\sigma_\ell(r_i)=w_i$ by contradiction.
	Two cases are possible.
	\begin{itemize}
		\item If $\sigma_\ell(r_i) < w_i$, the followers' cost by switching to resource $r_i$ satisfies
			$$
		 \frac{1}{w_i} (1-\sigma_\ell(r_i))+\frac{2s-\frac{1}{w_i}+1}{w_i}\sigma_\ell(r_i)<\frac{1}{w_{i}}(1-{w_i})+\frac{2s-\frac{1}{w_i}+1}{w_i} w_i=2s,
			$$
			where the inequality holds since the left-most quantity is a convex combination of $\frac{1}{w_i}$ and $\frac{2s - \frac{1}{w_i} + 1}{w_i}$ with weights $(1-\sigma_\ell(r_i))$ and $\sigma_\ell(r_i)$, and, since $\frac{1}{w_i} < \frac{2s - \frac{1}{w_i} + 1}{w_i}$, its maximum for $\sigma_\ell(r_i) \leq w_i$ is attained at $\sigma_\ell(r_i) = w_i$.
			Thus, we deduce that a follower would deviate from resource $r_{t_1}$ to resource $r_i$ (as her cost is $2s > 1$), contradicting the fact that $\nu$ is an NE for $\sigma_\ell$. 
		\item If $\sigma_\ell(r_i) > w_i$, we reach a contradiction since the cost incurred by the followers who are using resource $r_i$ would be $\frac{1}{w_i} \sigma_\ell(r_i) > 1$ and they would deviate playing resource $r_{t_1}$, decreasing their cost to $1$.
	\end{itemize}
	We have shown that $\sigma_\ell(r_i)=w_i$ for every resource $r_i$ with $\nu_{r_i}=2$.
	Finally, let $r_i$ be a resource with $\nu_{r_i} \neq 2$. 
	Clearly, it must be the case that $\sigma_\ell(r_i)=0$ since the leader's cost is $\epsilon$.
	Moreover, it cannot be the case that $\nu_{r_i}=1$, as, if it were the case, the follower would deviate to resource $r_{t_1}$ with a cost of 1, instead of $2s$.
	Similarly, $\nu_{r_i} \geq 3$ cannot hold, as one of the followers who are selecting resource $r_i$ would deviate playing $r_{t_1}$, since her current cost is greater than 1.
	Thus, either $\nu_{r_i}=2$ or $\nu_{r_i}=0$.
	As a consequence, there are $K$ resources $r_i$ with $\nu_{r_i}=2$ and $\sigma_\ell(r_i)=w_i$, and $|S|-K$ resources $r_i$ with $\nu_{r_i}=0$ and $\sigma_\ell(r_i)=0$.
	Let us define $S'$ as the set of integers $x_i \in S$ such that the corresponding resources $r_i$ satisfy $\nu_{r_i}=2$.
	Since $\sum_{x_i \in S'} \sigma_\ell(r_i)=1$ and $\sigma_\ell(r_i)=w_i$ for all such resources $r_i$, we can conclude that $\sum_{x_i \in S'} w_i=\sum_{x_i \in S'} \frac{x_i}{s}=1$, and, thus, $\sum_{x_i \in S'} x_i=s$.
	As a result, $(S',S \setminus S')$ is solution to $K$--PARTITION.
\end{proof}

Next, we show that even approximating the leader's cost in an OSE to within any polynomial factor of the input size is computationally hard, obtaining the same inapproximability result that we established for the problem of computing an OSE in the more general class of SSCGs.
For SSSCGs, the inapproximability result relies on the nonmonotonicity of the players' cost functions.
This must necessarily be the case since, as we will show in Section~\ref{sec:algorithms_same_actions_polynomial}, the problem is easy when all the action spaces are equal and the costs functions are monotonic.

\begin{theorem}
	The problem of computing an OSE in SSCGs is not in~\PolyAPX~unless \Poly = \NP.
\end{theorem}
\begin{proof}
	In order to prove the result, we rely on the reduction introduced in the proof of Theorem~\ref{thm:np_hard_opt_same}.
	We have already shown that in an OSE of $\Gamma_\epsilon(S,K)$ the leader's cost is $\epsilon$ if and only if the corresponding instance of $K$--PARTITION $(S,K)$ admits a \emph{yes} answer.
	Now, we prove that, when the $K$--PARTITION instance admits a \emph{no} answer, the leader's cost in any OSE is greater than or equal to 1.
	By contradiction, assume that there exists an OSE $ \sigma = (\sigma_\ell, \nu)$ in which the leader's cost is smaller than $1$. 
	Let $S' \subseteq S$ be the set of integers corresponding to a group of resources $r_i$ with $\nu_{r_i}=2$ (at least one must exist since the leader's cost is smaller than 1). 
	Then, $\sum_{x_i \in S'}\sigma_\ell(r_i) > \frac{s-1}{s}$ since $\sum_{x_i \in S \setminus S'} \sigma_\ell(r_i) + \sigma_\ell(r_{t_1})+\sigma_\ell(r_{t_2})$ must be smaller than~$\frac{1}{s}$ in order to have a leader's cost smaller than $1$.
	Moreover, $\sigma_\ell(r_{t_1}) \leq \frac{1}{s^4}$ and $\sigma_\ell(r_{t_1}) \leq \frac{1}{s^4}$ must both hold as, if not, the leader's cost would be larger than 1.
	We prove, now, that $\nu(r_{t_1})=4 |S|-2K+2$ by contradiction.
	We identify three cases:
	\begin{itemize}
		\item $\nu_{r_{t_1}} = 0$ implies that either there exists at least one resource $r_i$ with $\nu_{r_i} \ge 5$ or $\nu_{r_{t_2}}=2$, and, thus, either a follower who selected resource $r_i$ or one who selected resource $r_{t_2}$ would have an incentive to deviate to resource $r_{t_1}$ (as $4 s^2 > 3 s^2$).
		\item $1 \leq \nu_{r_{t_1}} \le 4 |S|-2 K-1$ implies that one of the followers who selected $r_{t_1}$ would have an incentive to deviate to resource $r_i$ with $\nu_{r_i}=2$, as she would incur a cost smaller than or equal to $\frac{2s-\frac{1}{w_i}+1}{w_i} < 3s^2$.
		\item $\nu_{r_{t_1}} = 4 |S|-2 K$ implies that the cost incurred by the followers who selected resource $r_{t_1}$ is greater or equal than $3s^2 (1-1/s^4)+2/s^3$, as $\sigma_\ell(r_{t_1}) \leq \frac{1}{s^4}$. Thus, since $\frac{2s-\frac{1}{w_i}+1}{w_i} < 2s^2 < 3s^2-3/s^2+2/s^3$, these followers would have an incentive to deviate from $r_{t_1}$ to a resource $r_i$ with $\nu_{r_i}=2$.
		%
		%
		\item $\nu_{r_{t_1}} \ge 4 |S|-2 K+2$ implies that the followers' cost after deviating to resource $r_{t_1}$ would be $0$ and, since there exists at least one resource $r_i$ with $\nu(r_i)=2$ and $\sigma(r_i)>0$, one of the followers who selected such resource would switch from it in favor of $r_{t_1}$.
	\end{itemize}
	%
	%
	%
	%
	%
	%
	Thus, $\nu(r_{t_1})=4 |S|-2 K+1$.
	Let us consider resource $r_{t_{t_2}}$.
	If $\nu_{r_{t_2}}=0$, the followers' cost incurred when deviating to resource $r_{t_2}$ would be smaller or equal than $(1-1/s^4) + 4/s^2$ (as $\sigma_\ell(r_{t_2}) \le 1/s^4$), while the cost incurred by choosing resource $r_{t_1}$ is at least $2s(1-1/s^4)+1/s^4 > (1-1/s^4) + 4/s^2$.
	Instead, if $\nu_{r_{t_2}}\ge 2$, the followers' cost for resource $r_{t_2}$ is $4s^2$ and they would have an incentive to deviate to $r_{t_1}$ to decrease their cost to $1$ or less.
	Thus, $\nu_{r_{t_2}}=1$.
	We deduce $\sigma_\ell(r_{t_1})=0$ as, otherwise (i.e., with $\sigma_\ell(r_{t_1})>0$), a follower would deviate from resource $r_{t_2}$ to $r_{t_1}$, decreasing her cost to $1$ or less.
	Let us focus on resources $r_i$ with $\nu_{r_i}=2$. 
	If $\sigma_\ell(r_i) < w_i$, the followers' cost of deviating to $r_i$ is 
	$$
	\frac{1}{w_i} (1-\sigma_\ell(r_i))+\frac{2s-\frac{1}{w_i}+1}{w_i}\sigma_\ell(r_i)<\frac{1}{w_{i}}(1-{w_i})+\frac{2s-\frac{1}{w_i}+1}{w_i} w_i=2s,
	 $$
	and they would deviate from $r_{t_1}$ to $r_i$, as their current cost is $2s$. 
	Instead, if $\sigma_\ell(r_i) > w_i$ the cost of any follower who selected $r_i$ is greater than $1$ and she would deviate to resource $r_{t_1}$ to decrease her cost to 1.
	Thus, $\sigma_\ell(r_i)=w_i$ for all resources $r_i$ with $\nu_{r_i}=2$.
	Now, let us consider a resource $r_i$ with $\nu(r_i) \neq 2$. 
	$\sigma_\ell(r_i) \leq \frac{1}{s}$ must hold, since the leader's cost is smaller than or equal to $1$.
	If $\nu(r_i)=1$, the followers' cost for resource $r_i$ is at least $2s \frac{1}{s}>1$ while, if $\nu(r_i)\ge 3$, the followers' cost for resource $r_i$ is at least $\frac{1}{w_i} > 1$.
	In both cases, the followers who selected resource $r_i$ would have an incentive to deviate to $r_{t_1}$ (as they would pay $1$).
	Thus, either $\nu_{r_i}=2$ or $\nu_{r_i}=0$.
	As a consequence, there are $K$ resources $r_i$ with $\nu_{r_i}=2$ and $\sigma_\ell(r_i)=w_i$ and $|S|-K$ resources $r_i$ with $\nu_{r_i}=0$.
	If the leader's cost for $\sigma$ is smaller than~$1$, there must be a subset $S'$ with $\sum_{x_i \in S'} \sigma_\ell(r_i) = \sum_{x_i \in S' }w_i > \frac{s-1}{s}$, which implies that $\sum_{x_i \in S'}\frac{x_i}{s} > \frac{s-1}{s}$ and $\sum_{x_i \in S'}x_i > s-1$. 
	Note that $x_i \in \mathbb{N}$ and $\sum_{x_i \in S'}x_i =s  \sum_{x_i \in S'}w_i = s \sum_{x_i \in S'}\sigma_\ell(r_i) \le s$. 
	Thus, $\sum_{x_i \in S'}x_i=s$ and $(S', S \setminus S')$ is solution to $K$--PARTITION.
	%
	So far, we have proven that $\Gamma_\epsilon(S,K)$ admits an OSE in which the leader's cost is $\epsilon$ if and only if $(S,K)$ has answer \emph{yes} and that, otherwise, the leader's cost is greater than or equal  $1$.
	Let $\epsilon = \frac{1}{2^{n+r}}$.
	Assume that there exists a polynomial-time approximation algorithm $\mathcal{A}$ with approximation factor $\text{poly}(n,r)$, i.e., a polynomial function of $n$ and $r$.
	Assume $(S,K)$ has answer \emph{yes}. $\mathcal{A}$ applied to $\Gamma_\epsilon(S,K)$ would return a solution with leader's cost at most $\frac{1}{2^{n+r}} \ \text{poly}(n,r)$. Since, for $n$ and $r$ large enough, $\frac{1}{2^{n+r}} \ \text{poly}(n,r) < 1$, $\mathcal{A}$ would allows us to decide in polynomial time whether $(S,K)$ has a \emph{yes} or \emph{no} answer, a contradiction unless \Poly $=$ \NP.
\end{proof}

\subsection{Computational Complexity of Finding a PSE in SSSCGs}\label{sub_sec:reduction_same_actions_pes}
	
We focus now on the problem of computing a PSE in SSSCGs.
%

%
%
\begin{theorem}\label{thm:np_hard_pes_same}
	Computing a PSE in SSSCGs is \NPHard.
\end{theorem}

\begin{proof}
	We provide a reduction from PARTITION showing that the existence of a polynomial-time algorithm for computing a PSE in SSSCGs would allow us to solve PARTITION in polynomial time.
	
	Let, as in the previous proof, $s= \frac{1}{2}\sum_{x_i \in S} x_i$
	%
	%
	and $w_i=\frac{x_i}{s}$ for all $x_i \in S$.
	W.l.o.g., let us assume that $x_i \leq s$ for all $x_i \in S$, and, thus, $w_i \leq 1$.
	Given a PARTITION instance with a set $S$, we build an instance $\Gamma_\epsilon(S)$ of an SSSCG with $0<\epsilon<1$ such that the leader's cost in a PSE is $\epsilon$ if and only if the PARTITION instance admits answer \emph{yes}.
	
	\textbf{Mapping.} 	
	$\Gamma_\epsilon(S)$ is defined as follows:
	\begin{itemize}
		\item $N = F \cup \{ \ell \}$, with $|F| = 3 |S|$;
		\item $R = \{ r_t\} \cup \{ r_i \mid i \in S \}$;
	\end{itemize}
        with the following cost functions:

        \bigskip

	\begin{center}
	    {\renewcommand{\arraystretch}{1}\begin{tabular}{c|cc|cc}
		\hline
		$x$ & $c_{r_i,f}$ & $c_{r_i,\ell}$  & $c_{r_t,f}$ & $c_{r_t,\ell}$   \\
		\cline{1-5}
		$1$ & $0$ & $\epsilon$ & $1$ & $s^4$   \\
		$ 2 $ & $\frac{1}{w_i-\frac{1}{s^4}}$ & $s^4$ & $1$ & $s^4$   \\
		$ 3 $ & $\frac{1}{1-w_i-\frac{1}{s^4}}$ & $\epsilon$ & $1$ & $s^4$  \\
		$ 4 $ & $0$ & $s^4$ & $1$ & $s^4$  \\
		$ [5,\infty] $ & $s$ & $\epsilon$ & $1$ & $s^4$  \\
	    \end{tabular}}            
        \end{center}

	
        \bigskip

	Clearly, $\Gamma_\epsilon(S)$ can be built in polynomial time, as it features $n=3|S|$ players, and $r=|S|+1$ resources.
	
	\textbf{If.} 
	Suppose that the PARTITION instance admits a \emph{yes} answer, and let $S' \subseteq S$ be such that $\sum_{x_i \in S'} x_i = s$.
	We show that there exists a PSE $\sigma = (\sigma_\ell,\nu)$ in which the leader's cost is $\epsilon$.
	Let $\sigma_\ell(r_i)=w_i $ for all $ x_i \in S'$, $\sigma_\ell(r_i)=0 $ for all $ x_i \notin S'$, and $\sigma_\ell(r_t)=0$. 
	We prove that the leader's cost is $\epsilon$ in any $\nu \in \mathbb{R}^r$ which is an NE in the followers' game induced by $\sigma_\ell$.
	Assume, by contradiction, that there exists an NE $ \nu$ in which the leader's cost is greater than $\epsilon$. 
	This implies that there exists a resource $r_i$ with $x_i \in S'$ and either $\nu_{r_i} =1$ or $\nu_{r_i} =3$.
	If $\nu_{r_i}=1$, the cost incurred by the followers who select $r_i$ is $\frac{1}{w_i-\frac{1}{s^4}}w_i>1$ and any of them would deviate to resource $r_t$ to decrease her cost to $1$. 
	If $\nu_{r_i}=3$, the followers' cost is $\frac{1}{1-w_i-\frac{1}{s^4}}(1-w_i)>1$ and any of them would deviate to resource $r_t$.
	In both cases, this contradicts the fact that $\nu $ is an NE, and, thus, it must be that $\nu(r_i)\neq 1$ and $\nu(r_i)\neq 3 $ for all $ x_i \in S'$.
	As a result, the leader's cost must be $\epsilon$ in any NE $\nu$.
	
	\textbf{Only if.}
	Suppose that $\Gamma_\epsilon(S)$ admits a PSE $\sigma = (\sigma_\ell,\nu)$ in which the leader's cost is $\epsilon$. 
	Then, $\sigma_\ell(r_{t})=0$ and $\sigma_\ell(r_i) > 0$ only if resource $r_i$ is such that $\nu_{r_i} \neq 1$ and $\nu_{r_i}\neq 3$.
	Let us define $R' \subseteq R$ as the set of resources $r_i$ with $\sigma_\ell(r_i)\le w_i-\frac{1}{s^4}$, $R''$ as the set of resources $r_i$ with $w_i-\frac{1}{s^4} < \sigma_\ell(r_i)<w_i+\frac{1}{s^4}$, and  $R'''$ as the set of resources $r_i$ with $\sigma_\ell(r_i) \ge w_i+\frac{1}{s^4}$. 
	Let $\nu \in \mathbb{R}^r$ be a followers' configuration such that $\nu_{r_i}=1 $ for all $ r_i \in R'$, $\nu_{r_i}=0 $ for all $ r_i \in R''$, $\nu_{r_i}=3 $ for all $ r_i \in R'''$, and $\nu_{r_t}=3|S|-\sum_{r_i \in R \setminus \{r_t\}} \nu_{r_i}$. 
	First, we show that $\nu$ is an NE for $\sigma_\ell$.
	Indeed, all the followers who selected resource $r_t$ incurs a cost of $1$, all those who selected a resource $r_i \in R'$ incur a cost of $\frac{1}{w_i-\frac{1}{s^4}} \sigma_\ell(r_i) < 1$, and all those who selected resource $r_i \in R'''$ incur a cost of $\frac{1}{1-w_i-\frac{1}{s^4}} (1 - \sigma_\ell(r_i)) < 1$. 
	%
	%
	 If any follower deviated, she would incur a cost greater or equal than $1$. 
	 In particular, no follower would deviate to a resource $r_{i} \in R'$, as she would incur a cost that is a convex combination of values greater than 1.
	 Similarly, no follower would deviate to a resource $r_{i} \in R''$ or $r_{i} \in R'''$, as she would incur a cost of, respectively, $\frac{1}{w_{i}-\frac{1}{s^4}} \sigma(r_{i})>1 $ or $s \sigma_\ell(r_{i})>1$.
	 Finally, no follower has an incentive to switch to resource $r_t$, as her cost would not increase.
	 This shows that, in the followers' game resulting from $\sigma_\ell$, there exists an NE such that, whenever the leader selects a resource $r_i$ in $R' \cup R''$, she incurs a cost of $s^4$.
	 Thus, given that the leader's cost in $\sigma$ is $\epsilon$, $R'=R'''=\emptyset$ must hold.
	 Let us define $S' \subseteq S$ as the set of integers $x_i \in S$ whose corresponding resource $r_i$ is such that $w_i-\frac{1}{s^4}<\sigma_\ell(r_i)<w_i+\frac{1}{s^4}$. 
	 For all the other resources $r_i$, it must be $\sigma_\ell(r_i)=0$.
	 Since $\sum_{x_i \in S'}\sigma_\ell(r_i)=1$, we have $\sum_{x_i \in S'}\left(w_i-\frac{1}{s^4}\right)<1<\sum_{x_i \in S'}\left(w_i+\frac{1}{s^4}\right)$, and, therefore, 
	 $$
	 \frac{s-1}{s}<1-\sum_{x_i \in S'} \frac{1}{s^4}<\sum_{x_i \in S'}w_i<1+\sum_{x_i \in S'} \frac{1}{s^4}<\frac{s+1}{s},
	 $$ 
	 which implies $s-1<\sum_{x_i \in S'}x_i<s+1$. 
	 Since $\sum_{x_i \in S'}x_i$ is an integer quantity, we deduce $\sum_{x_i \in S'}x_i=s$, which implies that $S'$ is a solution to PARTITION.
\end{proof}

Finally, we show that the same inapproximability result that we have established for OSEs also holds for PSEs.

\begin{theorem}
	The problem of computing an PSE in SSCGs is not in \PolyAPX~unless \Poly $=$ \NP.
\end{theorem}
\begin{proof}
	In order to prove the result, we rely on the reduction introduced in the proof of Theorem~\ref{thm:np_hard_pes_same}.
	We have already shown that in a PSE of $\Gamma_\epsilon(S)$ the leader's cost is $\epsilon$ if and only if the corresponding instance of PARTITION admits a \emph{yes} answer.
	Now, we show that, if the partition problem has \emph{no} answer, then the leader's cost in any PSE is greater than or equal to 1.
	Suppose, by contradiction, that there is a leader's strategy $\sigma_\ell$ such that all NEs of the resulting followers' game provide the leader with a cost smaller than $1$.
	Then, $\sigma_\ell(r_i) < \frac{1}{s^4}$ for all resources $r_i$ such that $\nu_{r_i} = 3$, $\sigma_\ell(r_i) < \frac{1}{s^4}$ for all resources $r_i$ such that $\nu_{r_i} = 1$, and $\sigma_\ell(r_t) < \frac{1}{s^4}$.
	If there is a resource $r_i$ with $\sigma_\ell(r_i)>w_i+\frac{1}{s^4}$, we have already proven that there is an NE with $\nu_{r_i}=3$ providing the leader with a cost greater than $s^4 \sigma_\ell(r_i) >1$.
	Consider the set $S'' \subseteq S$ of integers $x_i$ corresponding to resources $r_i$ with $\sigma_\ell(r_i) \le w_i-\frac{1}{s^4}$. We have already shown that there is an NE with $\nu_{r_i}=1 $ for all $ x_i \in S''$. 
	%
	%
	Since the leader can select these resources with, at most, probability $\frac{1}{s^4}$ (as $\sum_{x_i \in S''} \sigma_\ell(r_i) \le \frac{1}{s^4}$), there is a set $S'$ of resources $r_i$ with $w_i-\frac{1}{s^4}<\sigma_\ell(r_i)<w_i+\frac{1}{s^4}$ and $\sum_{x_i \in S'} \sigma_\ell(r_i)\ge1-\frac{1}{s^4}$.
	From $\sum_{x_i \in S'}(w_i+\frac{1}{s^4})>\sum_{x_i \in S'}\sigma_\ell(r_i)\ge1-\frac{1}{s^4}$, we obtain $\sum_{x_i \in S'}w_i>1-\frac{1}{s^4}-\frac{|S'|}{s^4}>\frac{s-1}{s}$.
        From $\sum_{x_i \in S'}(w_i-\frac{1}{s^4})<\sum_{x_i \in S'}\sigma_\ell(r_i)\le 1$, we deduce $\sum_{x_i \in S'}w_i <1+\frac{|S'|}{s^4}<\frac{s+1}{s}$.
	Thus, $ s-1<\sum_{x_i \in S'}x_i<s+1$ and, since $\sum_{x_i \in S'}x_i$ is an integer quantity, we have that $\sum_{x_i \in S'}x_i=s$, showing that $S'$ is a solution to PARTITION.
	We have proven that $\Gamma_\epsilon(S)$ admits a PSE in which the leader's cost is $\epsilon$ if and only if the PARTITION instance has a \emph{yes} answer, while, otherwise, the leader's cost is greater than or equal to $1$.
	Let $\epsilon = \frac{1}{2^{n+r}}$.
	Assume that there exists a polynomial-time approximation algorithm $\mathcal{A}$ with approximation factor $\text{poly}(n,r)$, i.e., a polynomial function of $n$ and $r$.
	Assume the PARTITION instance has a answer \emph{yes}.
        $\mathcal{A}$ applied to $\Gamma_\epsilon(S)$ would return a solution with leader's cost at most $\frac{1}{2^{n+r}} \ \text{poly}(n,r)$. Since, for $n$ and $r$ large enough, $\frac{1}{2^{n+r}} \ \text{poly}(n,r) < 1$, $\mathcal{A}$ would allow us to decide in polynomial time whether the PARTITION instance has a \emph{yes} or \emph{no} answer, a contradiction unless $\mathsf{P} = \mathsf{NP}$.
\end{proof}

\section{Polynomial-Time Algorithms for SSSCGs}\label{sec:algorithms_same_actions_polynomial}

In the previous sections, we have shown that the problem of computing an O/PSE in SSCGs is, both in the general case and when restricting ourselves to SSSCGs, computationally intractable.
We provide, here, two positive results for SSSCGs, showing that, under certain conditions, the computation of an O/PSE in these games can be carried out in polynomial time.

First, we design a polynomial-time algorithm for finding an O/PSE in SSSCGs where the players' costs are monotonic functions of the resource congestion.
The algorithm relies on the fact that, as we will show, in such games the leader cannot decrease her cost by playing mixed strategies and, thus, pure-strategy commitments are sufficient.
We also exhibit a few examples showing that our algorithm cannot be easily extended to more general settings as, if the players have either nonidentical action spaces or nonmonotonic cost functions, the leader could be better off playing mixed strategies, thus violating the fundamental assumption of our algorithm.

Finally, we show that, if we restrict our attention to pure-strategy commitments in SSSCGs, an O/PSE can be found in polynomial time by means of a \emph{dynamic programming} (DP) algorithm, even when the players' cost functions are nonmonotonic.

\subsection{Polynomial-time algorithms for computing an O/PSE in SSSCGs with Monotonic Cost Functions}\label{sub_sec:greedy}

Let us recall that, in SSSCGs, an NE minimizing the social cost can be computed in polynomial time~\cite{ieong2005fast}.
It is also easy to show that an NE minimizing/maximizing the cost incurred by one player can be found efficiently, using an algorithm similar to that of~\cite{ieong2005fast} (see Section~\ref{sub_sec:dynamic_programming} for additional details).
As a consequence, computing an O/PSE would also be easy if, in the followers' game, an NE could only be induced by a leader's commitment in pure strategies.
This is, unfortunately, not the case, as the following examples shows:
%
%
%
\begin{proposition}
	There are SSSCGs with weakly monotonic cost functions where some followers configurations are NEs only for a mixed-strategy commitment of the leader.
\end{proposition}
\begin{proof}
%
%
Consider the following game with weakly monotonic cost functions, where $|F|=3$ and $R = \{r_1,r_2,r_3\}$.
\begin{center}
  \renewcommand{\arraystretch}{1}\begin{tabular}{c|cc|cc|cc}
    \hline
    $x$ & $c_{r_1,\ell}$ & $c_{r_1,f}$ & $c_{r_2,\ell}$ & $c_{r_2,f}$ & $c_{r_3,\ell}$ & $c_{r_3,f}$ \\
    \cline{1-7}
    $1$ & $1$ & $1$ & $3$ & $4$ & $1$ & $1$  \\
    $ 2 $ & $2$ & $3$ & $4$ & $5$ & $2$ & $3$ \\
    $ 3 $ & $3$ & $6$ & $5$ & $6$ & $3$ & $6$ \\
    \end{tabular}
\end{center}
The followers configuration $\nu = (1,1,1)^T$ in which each follower selects a different resource is not an NE if the leader commits to a pure strategy, while, for instance, it is an NE for $\sigma_\ell(r_1) = \sigma_\ell(r_3) = \frac{1}{2}$ and $\sigma_\ell(r_2) = 0$.
Moreover, notice that the game admits O/PSEs in which the leader's commitment is a mixed strategy.
For instance, for $\sigma_\ell(r_1) = \sigma_\ell(r_3) = \frac{1}{2}$ and $\sigma_\ell(r_2) = 0$, the leader incurs a cost of $2$, and there is no other strategy that allows her to pay less than 2.
\end{proof}

We now show that, when searching for an OSE in SSSCGs with weakly monotonic cost functions, one can w.l.o.g. restrict the attention to pure-strategy commitments of the leader as, given any OSE in which the leader plays a mixed strategy, one can easily construct another equilibrium in which, instead, the leader's strategy is pure.


\begin{theorem}\label{thm:pure_strategy_opt}
	Every SSSCG with weakly monotonic cost functions admits an OSE $\sigma = (\sigma_\ell,\nu)$ in which $\sigma_\ell$ is pure.
\end{theorem}
\begin{proof}
	Given an OSE $\sigma = (\sigma_\ell,\nu)$ with $\sigma_\ell$ mixed, we show how to construct another OSE $\hat \sigma = (\hat \sigma_\ell, \hat \nu)$ where $\hat \sigma_\ell$ is pure.
	Let $S = \{ i \in R \mid \sigma_\ell(i) > 0 \}$ be the set of resources played by the leader with positive probability in $\sigma_\ell$, and let $i^\star \in \arg\min_{i \in S} c_{i,\ell}(\nu_i + 1)$.
	Clearly, since the leader's utility is a convex combination, weighted by $\sigma_\ell$, of the costs she incurs in the resources chosen with positive probability, $c_\ell^{\sigma} = \sum_{i \in A_\ell} \sigma_\ell(i) c_{i,\ell}(\nu_i+1) \geq c_{i^\star,\ell}(\nu_{i^\star} + 1)$.
	Moreover, since $\nu$ is an NE for $\sigma_\ell$, the following holds:
	\begin{equation}\label{eq:starting_ne}
	c_{i, f}^{\sigma_\ell}(\nu_i) \leq c_{j, f}^{\sigma_\ell}(\nu_j + 1) \ \ \forall\ i \in R : \nu_i > 0, \ j \in R.
	\end{equation}
	Let us define $\hat \sigma_\ell$ such that $\hat \sigma_\ell(i^\star) = 1$. We now show that such $\hat \sigma_\ell$ is part of an OSE.
	Notice that $c_{i, f}^{\hat \sigma_\ell}(x) = c_{i, f}(x) \ \forall x \in \mathbb{N}$ for every $i \in R \setminus\{i^\star\}$ (as the leader does not select these resources), while $c_{i^\star, f}^{\hat \sigma_\ell}(x) = c_{i^\star, f}(x+1) \ \forall x \in \mathbb{N}$ (as the leader selects that resource).
	Since the followers behave optimistically, it is sufficient to exhibit a $\hat \nu \in E^{\hat \sigma_\ell}$ such that $\hat \sigma = (\hat \sigma_\ell, \hat \nu)$ satisfies $c_\ell^{\hat \sigma} \leq c_\ell^{\sigma}$.
	We construct a sequence of followers configurations reaching such $\hat \nu$.
	Given $\hat \sigma_\ell$, let us consider the sequence $(\nu(0)=\nu,\nu(1),\ldots,\nu(T)=\hat \nu)$ such that each configuration differs from the previous one in that a single follower has changed resource, strictly decreasing her cost.
	%
	Formally, this corresponds to showing that, for all $0 \leq t < T$, there is a pair $i,j \in R$ such that $\nu(t)_i > 0$, $\nu(t+1)_i = \nu(t)_i -1$, $\nu(t+1)_j = \nu(t)_j +1$, and $c_{i,f}^{\hat \sigma_\ell} (\nu(t)_i) > c_{j,f}^{\hat \sigma_\ell} (\nu(t+1)_j)$.
	Moreover, let us assume that a follower deviates to resource $i^\star$, i.e., $\nu(t+1)_{i^\star} > \nu(t)_{i^\star}$, only if this is the only way of strictly decreasing some follower's cost. This is w.l.o.g., as it is consistent with the assumption of optimism.
	Let us now prove the following:
	\begin{equation}\label{eq:non_increasing}
	\nu(t+1)_{i^\star} \leq \nu(t)_{i^\star} \ \ \forall\ 0 \leq t < T.
	\end{equation}
	By contradiction, assume there exists $0 \leq t < T$ such that $\nu(t+1)_{i^\star} > \nu(t)_{i^\star}$.
	Then, there is a follower in $\nu(t)$ who can strictly decrease her cost by choosing $i^\star$ instead of some resource $j \neq i^\star \in R : \nu(t)_j > 0$.
	Thus:
	\begin{equation}\label{eq:sequence_dev}
	c_{i^\star, f}^{\sigma_\ell}(\nu_{i^\star} + 1) \leq c_{i^\star, f}(\nu(t)_{i^\star} + 2) < c_{j,f}(\nu(t)_j),
	\end{equation}
	where the first inequality holds since $\nu(t)_{i^\star} = \nu_{i^\star} $.
	Two cases are possible.
	In the first one, $\nu(t)_j \leq \nu_j$, implying $c_{j,f}(\nu(t)_j) \leq c_{j,f}(\nu_j) \leq c_{j,f}^{\sigma_\ell}(\nu_j)$, which, together with Equations~\eqref{eq:starting_ne}~and~\eqref{eq:sequence_dev}, leads to a contradiction.
	In the second case, $\nu(t)_j > \nu_j$ implies that there exists $k \neq i^\star \in R$ such that $\nu(t)_k < \nu_k$ (and $\nu_k > 0$), otherwise $\sum_{i \in R} \nu(t)_i > n - 1$.
	It follows that $c_{j,f}(\nu(t)_j) \leq c_{k,f}(\nu(t)_k + 1)  \leq c_{k,f}^{\sigma_\ell}(\nu_k)$, where the first inequality holds since, due to our assumptions on the sequence, it cannot be $c_{j,f}(\nu(t)_j) > c_{k,f}(\nu(t)_k + 1)$ as $\nu(t+1)_{i^\star} > \nu(t)_{i^\star}$, and the second inequality follows from $\nu(t)_k < \nu_k$.
	Thus, Equations~\eqref{eq:starting_ne}~and~\eqref{eq:sequence_dev} give a contradiction.
	As a result, Equation~\eqref{eq:non_increasing} holds, and, thus, $\hat \nu_{i^\star} \leq \nu_{i^\star}$.
	Given the monotonicity of the costs, $\hat \sigma$ is an OSE.
\end{proof}

We prove, now, that a similar result holds for the pessimistic case, i.e., for computing a PSE.
The result is weaker though, as it requires the stronger assumption that the followers' cost functions be strictly monotonic.

\begin{theorem}\label{thm_pure_strategies_pes}
	Every SSSCG in which the leader's and followers' cost functions are, respectively, weakly and strictly monotonic, admits a PSE $\sigma = (\sigma_\ell,\nu)$ in which $\sigma_\ell$ is pure.
\end{theorem}

\begin{proof}
	Assume there exists a PSE $\sigma = (\sigma_\ell, \nu)$ in which $\sigma_\ell$ is mixed.
	We show that there must be another PSE $\hat \sigma = (\hat \sigma_\ell, \hat \nu)$ such that $\hat \sigma_\ell$ is pure.
	Let us define $i^\star \in R$ as in the proof of Theorem~\ref{thm:pure_strategy_opt}, so that $c_\ell^{\sigma} \geq c_{i^\star,\ell}(\nu_{i^\star} + 1)$ and Equation~\eqref{eq:starting_ne} holds.
	Given that the followers behave pessimistically, we need to show that, for every $\hat \nu \in E^{\hat \sigma_\ell}$, $\hat \sigma = (\hat \sigma_\ell, \hat \nu)$ satisfies $c_\ell^{\hat \sigma} \leq c_\ell^{\sigma}$.
	By contradiction, assume $c_\ell^{\hat \sigma} > c_\ell^\sigma$, which implies $c_{i^\star,\ell}(\hat \nu_{i^\star} + 1) > c_{i^\star,\ell}(\nu_{i^\star} + 1)$.
	It easily follows from the monotonicity of the costs that $\hat \nu_{i^\star} > \nu_{i^\star}$.
	Thus, there must be a resource $ j \in R$ such that $\hat \nu_{ j} < \nu_{ j}$ as, otherwise, $\sum_{i \in R} \hat \nu_i > n - 1$.
	Let us also remark that $\nu_{ j} > 0$.
	Thus:
	\begin{equation}\label{eq:pes_contradiction}
	\hspace{-0.05cm} c_{i^\star, f}^{\sigma_\ell}(\nu_{i^\star} + 1) \leq c_{i^\star, f}(\hat \nu_{i^\star} + 1) \leq c_{ j, f}(\hat \nu_{ j} + 1) \leq c_{ j, f}^{\sigma_\ell}(\nu_{ j}),
	\end{equation}
	where the first inequality follows from $\nu_{i^\star} < \hat \nu_{i^\star}$, the second one from the fact that $\hat \nu$ is an NE for $\hat \sigma_\ell$, and the third one from $\hat \nu_{ j} < \nu_{ j}$.
	Equation~\eqref{eq:starting_ne} implies $c_{ j, f}^{\sigma_\ell}(\nu_{ j}) \leq c_{i^\star,f}^{\sigma_\ell}(\nu_{i^\star}+ 1)$.
	If $c_{ j, f}^{\sigma_\ell}(\nu_{ j}) < c_{i^\star,f}^{\sigma_\ell}(\nu_{i^\star} + 1)$, then Equation~\eqref{eq:pes_contradiction} leads to a contradiction.
	Otherwise, if $c_{j, f}^{\sigma_\ell}(\nu_{ j}) = c_{i^\star,f}^{\sigma_\ell}(\nu_{i^\star} + 1)$ all the inequalities in Equation~\eqref{eq:pes_contradiction} hold as equations.
	This, however, implies $c_{i^\star, f}^{\sigma_\ell}(\nu_{i^\star} + 1) = c_{i^\star, f}(\hat \nu_{i^\star} + 1) $ and $c_{ j, f}(\hat \nu_{ j} + 1) = c_{ j, f}^{\sigma_\ell}(\nu_{ j})$, which is a contradiction since $\sigma_\ell$ is mixed and the followers' cost functions are strictly monotonic. 
\end{proof}

Theorem~\ref{thm_pure_strategies_pes} fails to hold if the followers' cost functions are weakly, rather than strongly, monotonic, as the following result shows:
\begin{proposition}
	There are SSSCGs with weakly monotonic cost functions where any PSE prescribes the leader to play a mixed strategy.
\end{proposition}
\begin{proof}
Consider the following instance of an SSSCG with weakly monotonic cost functions, where $|F|=1$ and $R = \{r_1,r_2\}$.
%

\begin{center}
  \renewcommand{\arraystretch}{1}\setlength{\tabcolsep}{3pt}\begin{tabular}{c|cc|cc}
    \hline
    $x$ & $c_{r_1,\ell}$ & $c_{r_1,f}$  & $c_{r_2,\ell}$ & $c_{r_2,f}$ \\
    \cline{1-5}
    $ 1 $ & $1$ & $1$ & $1$ & $1$ \\
    $ 2 $ & $2$ & $1$ & $2$ & $1$ \\
    \end{tabular}
\end{center}

Clearly, any followers' configuration is an NE in this game, independently of the leader's commitment.
Whenever the leader commits to a pure strategy, be it the selection of $r_1$ or $r_2$, the follower, due to the pessimistic assumption, chooses the same resource, so to have the leader incur a cost as large as possible (of $2$).
By uniformly randomizing between the two resources, though, the leader can reduce her cost to $2 \frac{1}{2} + \frac{1}{2} = 1.5$.
\end{proof}

%
Relying on Theorems~\ref{thm:pure_strategy_opt}~and~\ref{thm_pure_strategies_pes}, we can compute an OSE (respectively, PSE) by enumerating the leader's pure strategies and, for each of them, computing a followers' NE which results in the smallest (respectively, largest) leader's cost. 
Such NE can be computed by applying a simple greedy procedure that progressively assigns followers to resources.
At each step, a single follower is assigned to the resource which is cheapest for her, given how previously the considered followers have been distributed over the resources.
Moreover, at a given step, among all the resources minimizing followers' cost, the procedure selects one minimizing (respectively, maximizing) the leader's cost.
An O/PSE is then obtained by picking any leader's pure strategy for which the leader's cost is the smallest.

The detailed procedure is described in Algorithm~\ref{algorithm}, where, for some $S \subseteq R$ and $i \in S$, the function \texttt{O-Pick}$(S,i)$ (respectively, \texttt{P-Pick}$(S,i)$) returns some resource $j^\star \in S$, giving precedence to resources $j^\star \neq i$ (respectively, $j^\star = i$).

\begin{algorithm}[h!tp]
	
	\SetKwProg{Fn}{Function}{}{}
	\SetKwFunction{AlgName}{Compute-O/P-LFE}
	\SetKwFunction{Pick}{O/P-Pick}
	\SetKwInOut{Input}{input}\SetKwInOut{Output}{output}
	
	\Input{An SSSCG $\Gamma = (N,R,c_\ell,c_f)$}
	\Output{$\sigma$ that is an O/P-LFE of $\Gamma$}
	\BlankLine
	
	\Fn{\AlgName{$\Gamma$}}{
		
		\For{$i \in R$}{
			$\sigma_\ell[i] \leftarrow \sigma_\ell \in \Delta_\ell : \sigma_\ell(i) = 1$\;
			$\nu[i,j] \leftarrow 0 \ \ \forall\ i,j \in R$\;
			\While{$\sum_{j \in R} \nu[i,j] < n$}{
				$S \leftarrow \arg\min_{j \in R} \ \ c_{j,f}^{\sigma_\ell[i]} (\nu[i,j] + 1) $\;
				$j^\star \leftarrow $ \Pick{$S,i$}\;
				$\nu[i,j^\star]\leftarrow \nu[i,j^\star] + 1$\;
			}
			$c_\ell[i] \leftarrow c_{i,\ell}(\nu[i,i] + 1)$\;
		}
		$i^\star \leftarrow \arg\min_{i \in R} c_\ell[i]$\;
		\Return $\sigma = (\sigma_\ell[i^\star], \nu[i^\star,\cdot])$\;
		
	}
	
	\caption{Algorithm computing an O/PSE of an SSSCG.}
	\label{algorithm}
\end{algorithm}

Let us remark that, in Algorithm~\ref{algorithm}, $\sigma_\ell[\cdot]$, $\nu[\cdot,\cdot]$, and $c_\ell[\cdot]$ are the algorithm's variables, and, for every $i \in R$, $\nu[i,j]$ denotes the number of followers selecting resource $j \in R$ in the NE that is reached when the leader's strategy is $\sigma_\ell[i]$.

\begin{theorem}\label{thm:algorithm}
	Algorithm~\ref{algorithm} is correct and it runs in time $O(n r \log r)$.
\end{theorem}

\begin{proof}
  We rely on the pseudocode reported in Algorithm~\ref{algorithm} to show its correctness.
  Thanks to Theorems~\ref{thm:pure_strategy_opt}~and~\ref{thm_pure_strategies_pes}, we only need to prove that, for every $i \in R$ and after the execution of the while loop, the followers configuration $\nu$ is such that, for all $j \in R$, $\nu_j = \nu[i,j]$ is an NE for $\sigma_\ell[i]$ minimizing (or maximizing) the leader's cost.
	First, let us show that $\nu$ is an NE. 
	Suppose, by contradiction, that it is not.
	Then, there exists $j \in R : \nu_j > 0$ and $k \in R$ such that $c_{j,f}^{\sigma_\ell[i]} (\nu_j) > c_{k,f}^{\sigma_\ell[i]} (\nu_k + 1)$.
	Let $\bar \nu_k$ be the value of $\nu[i,k]$ during the step in which $\nu[i,j]$ is set to its final value $\nu_j$.
	Clearly, $c_{j,f}^{\sigma_\ell[i]} (\nu_j) > c_{k,f}^{\sigma_\ell[i]} (\nu_k + 1) \geq c_{k,f}^{\sigma_\ell[i]} (\bar \nu_k + 1)$, and the algorithm would have not incremented $\nu[i,j]$ during that step, a contradiction.
	Let us show now that $(\sigma_\ell[i], \nu)$ is an O/PSE.
        In the remainder of the proof, we focus on the optimistic case (the pessimistic one can be treated analogously).
	Suppose, by contradiction, that $\nu$ is not an NE minimizing the leader's cost for $\sigma_\ell[i]$ (i.e., not an OSE).
	Then, there exists another NE $\hat \nu$ for $\sigma_\ell[i]$ such that $c_{i,\ell}(\hat \nu_i + 1) < c_{i,\ell}(\nu_i + 1)$.
	Given the monotonicity of the costs, $\hat \nu_i < \nu_i$ must hold.
	Therefore, there must exist some $j \neq i \in R$ such that $\hat \nu_j > \nu_j$. 
	Let us consider the step in which $\nu[i,i]$ is set to $\nu_i$, and let $\bar \nu_j$ be the value of $\nu[i,j]$ during that step.
	Note that $c_{i,f}^{\sigma_\ell[i]}(\nu_i) < c_{j,f}^{\sigma_\ell[i]}(\bar \nu_j + 1)$ must hold as, otherwise, the algorithm would have incremented $\nu[i,j]$ instead of $\nu[i,i]$.
	But, then, $c_{j,f}^{\sigma_\ell[i]}(\bar \nu_j +1 ) \leq c_{j,f}^{\sigma_\ell[i]}(\nu_j +1) \leq c_{j,f}^{\sigma_\ell[i]}(\hat \nu_j)$, which implies $c_{i,f}^{\sigma_\ell[i]}(\hat \nu_i + 1) \leq c_{i,f}^{\sigma_\ell[i]}(\nu_i) < c_{j,f}^{\sigma_\ell[i]}(\bar \nu_j + 1) \leq c_{j,f}^{\sigma_\ell[i]}(\hat \nu_j)$, contradicting the fact that $\hat \nu$ is an NE for the given $\sigma_\ell[i]$.

	Since the while loop is executed exactly $r$ times, each execution carries out $n$ steps.
	Using efficient data structures, 
        each step takes time $O(\log r)$.
	Thus, the overall running time is $O(nr \log r)$.
\end{proof}

Next, we provide a characterization of O/PSEs in SSSCGs with monotonic costs under the additional assumption that leader's and followers' costs be equal, which may be of independent interest besides the computation of such equilibria:
%
%
\begin{theorem}\label{thm:ne_charaterization}
	Given an SSSCG with monotonic costs and $c_\ell =c_f = \{ c_i \}_{i\in R}$, any O/PSE $\sigma = (\sigma_\ell, a)$ with $\sigma_\ell$ pure is an NE.
\end{theorem}
\begin{proof}
	Let $\sigma = (\sigma_\ell, \nu)$ be an O/PSE with $\sigma_\ell(i^\star) = 1$ for some $i^\star \in R$.
	Clearly, given that $\nu \in E^{\sigma_\ell}$, $c_i^{\sigma_\ell}(\nu_i) \leq c_j^{\sigma_\ell}(\nu_j + 1)$ holds for every $i \in R : \nu_i > 0$ and for every $j \in R$. 
	Therefore, no follower has an incentive to change resource. Thus, it is sufficient to prove that the leader has no incentive to deviate from resource $i^\star$ unilaterally, i.e., without assuming that the followers would react to her deviation (which is the case in the Stackelberg setting).
	If $\nu_{i^\star} > 0$, we have $c_{i^\star}(\nu_{i^\star} + 1) = c_{i^\star}^{\sigma_\ell}(\nu_{i^\star}) \leq c_j^{\sigma_\ell}(\nu_j + 1) = c_j(\nu_j + 1)$ for every $j\neq i^\star \in R$, and it immediately follows that the leader does not deviate and $\sigma$ is an NE.
	The case in which $\nu_{i^\star} = 0$ is more involved.
	By contradiction, assume that $\sigma$ is not an NE.
	As a consequence, the leader must have an incentive to deviate to some resource $j \neq i^\star \in R$, i.e., $c_{i^\star}( \nu_{i^\star} + 1) = c_{i^\star}(1) > c_j(\nu_j + 1)$.
	Let $\hat \sigma_\ell$ with $\hat \sigma_\ell(j) = 1$ be the strategy the leader commits to.
	We prove (by contradiction) that, for every $\hat \nu \in E^{\hat \sigma_\ell}$, $\hat \sigma = (\hat \sigma_\ell, \hat \nu)$ provides the leader with a cost strictly smaller than $c_{i^\star}(1)$. 
	Assume $c_j(\hat \nu_j + 1) \geq c_{i^\star}(1)$. 
	Three cases are possible.
	In the first one, $\hat \nu_j < \nu_j$ and $c_{i^\star}(1) > c_j(\nu_j + 1) \geq c_j(\hat \nu_j + 1) \geq c_{i^\star}(1)$.
	In the second one, $\hat \nu_j = \nu_j$ and $c_j(\hat \nu_j + 1) \geq c_{i^\star}(1) > c_j(\nu_j + 1)$.
	In the third case, $\hat \nu_j > \nu_j$, which implies that there must be a resource $k \neq i^\star \in R$ such that $\hat \nu_k < \nu_k$, and $c_{i^\star}(1) > c_j(\nu_j + 1) \geq c_k(\nu_k) \geq c_k(\hat \nu_k +1 ) \geq c_j(\hat \nu_j + 1) \geq c_{i^\star}(1)$.
	As all the cases lead to a contradiction, it must be $c_j(\hat \nu_j + 1) < c_{i^\star}(1)$.
	The proof is complete as, in $\hat \sigma$, the leader's cost is $c_j(\hat \nu_j + 1) < c_{i^\star}(1)$, contradicting the fact that $\sigma$ is an O/PSE.
\end{proof}

\subsection{On the Necessity of the Assumptions We Made}

We provide some examples showing why Algorithm~\ref{algorithm} cannot be easily extended to more general settings---the reason being that Theorems~\ref{thm:pure_strategy_opt}~and~\ref{thm_pure_strategies_pes} do not hold if the assumption of monotonicity is dropped.

First, let us analyze the general case of SSSCGs in which the costs need not be monotonic functions of the resource congestion:
\begin{proposition}
There are SSSCGs in which, even if the cost functions of one player only are nonmonotonic, be it the leader or one of the followers, any O/PSE prescribes the leader to play a mixed strategy.
\end{proposition}
%
%
%
\begin{proof}
Consider the following SSSCG with $R = \{r_1,r_2\}$ and a single follower ($|F|=1$) with nonmonotonic cost functions:

\begin{center}
  \renewcommand{\arraystretch}{1}
  \setlength{\tabcolsep}{3pt}\begin{tabular}{c|cc|cc}
    \hline
    $x$ & $c_{r_1,\ell}$ & $c_{r_1,f}$  & $c_{r_2,\ell}$ & $c_{r_2,f}$ \\
    \cline{1-5}
    $ 1 $ & $1$ & $2$ & $1$ & $2$ \\
    $ 2 $ & $2$ & $1$ & $2$ & $1$ \\
    \end{tabular}
\end{center}

\noindent The follower selects $r_2$ whenever $\sigma_\ell(r_1) \leq \frac{1}{2}$, while, if $\sigma_\ell(r_1) \geq \frac{1}{2}$, she chooses $r_1$. The leader's cost is $2 - \sigma_\ell(r_1)$ if $\sigma_\ell(r_1) \leq \frac{1}{2}$, and  $1 + \sigma_\ell(r_1)$ if $\sigma_\ell(r_1) \geq \frac{1}{2}$. 
%
There is, thus, a unique O/PSE that prescribes the leader to commit to $\sigma_\ell$ with $\sigma_\ell(r_1) = \sigma_\ell(r_2) = \frac{1}{2}$.

Consider now the following SSSCG with $R = \{r_1,r_2\}$ and single follower ($|F|=1$), with nonmonotonic leader cost functions:

\begin{center}
  \renewcommand{\arraystretch}{1}\setlength{\tabcolsep}{3pt}
  \begin{tabular}{c|cc|cc}
    \hline
    $x$ & $c_{r_1,\ell}$ & $c_{r_1,f}$  & $c_{r_2,\ell}$ & $c_{r_2,f}$ \\
    \cline{1-5}
    $ 1 $ & $2$ & $1$ & $2$ & $1$ \\
    $ 2 $ & $0$ & $2$ & $0$ & $2$ \\
  \end{tabular}
\end{center}

\noindent The follower selects $r_2$ if $\sigma_\ell(r_1) \geq \frac{1}{2}$, and $r_1$ if $\sigma_\ell(r_1) \leq \frac{1}{2}$.
The leader's cost is thus $2 \sigma_\ell(r_1)$ if $\sigma_\ell(r_1) \geq \frac{1}{2}$, and $2 - 2 \sigma_\ell(r_1)$ if $\sigma_\ell(r_1) \leq \frac{1}{2}$.
%
There is, thus, a unique O/PSE that prescribes the leader to commit to $\sigma_\ell$ with $\sigma_\ell(r_1) = \sigma_\ell(r_2) = \frac{1}{2}$.
\end{proof}

Finally, we show that Theorems~\ref{thm:pure_strategy_opt}~and~\ref{thm_pure_strategies_pes} do not hold for SSCGs even in the extreme case where all the cost functions are monotonic:
\begin{proposition}
	There are SSCGs with monotonic cost functions where any O/PSE prescribes the leader to play a mixed strategy.
\end{proposition}
\begin{proof}

Consider the following SSCG with $R = \{r_1,r_2,r_3\}$, two followers $F = \{p_1,p_2\}$, and $A_{p1} = \{r_1,r_2\}, A_{p2} = \{r_2,r_3\}, A_\ell=\{r_1,r_2\}$:

\begin{center}
  \renewcommand{\arraystretch}{1}\begin{tabular}{c|cc|cc|c}
    \hline
    $x$ & $c_{r_1,f}$ & $c_{r_1,\ell}$  & $c_{r_2,f}$ & $c_{r_2,\ell}$ & $c_{r_3,f}$ \\
    \cline{1-6}
    $1$ & $1$ & $0$ & $0$  & $1$& $3$ \\
    $ 2 $ & $1$ & $1$ & $2$ &$1$ &$3$\\
    $ 3 $ & $1$ & $1$ & $4$& $1$ &$3$\\
    \end{tabular}
\end{center}

If the leader plays $\sigma_\ell(r_1) = 1$,
there is a unique NE where follower $p_1$ plays $r_1$ and follower $p_2$ plays $r_2$. Indeed, $p_2$ incurs a cost of 0 and, thus, has no incentive to deviate, while $p_1$ would incur a cost of $2 > 1$ by deviating to $r_2$.
Thus, the leader's cost is $1$.
The leader's cost is also $1$ if she played $\sigma_\ell(r_2) = 1$, as $p_2$ would also choose $r_2$, while $p_1$ would choose $r_1$.

Let us show that the leader can commit to a mixed strategy and incur a cost smaller than $1$.
Indeed, with $\sigma_\ell(r_1) = \sigma_\ell(r_2) = \frac{1}{2}$, there is a followers' NE where $p_1$ chooses $r_2$ and $p_2$ chooses $r_3$: $p_1$, incurring a cost of $1$ (smaller or equal than any other cost), has no incentive to deviate, while $p_2$, currently incurring a cost of $3$, by switching to $r_2$ would incur  the same (expected) cost of $3$ (i.e., a cost of $2$ with probability $\frac{1}{2}$ and one of $4$ with probability $\frac{1}{2}$), thus having no incentive to deviate.
At that NE, the leader's cost is $0 \cdot \frac{1}{2} +  1 \cdot \frac{1}{2} = \frac{1}{2}$.
\end{proof}

\subsection{Pure-Strategy Commitment in SSSCGs with Arbitrary Costs}\label{sub_sec:dynamic_programming}

We propose, here, a simple polynomial-time algorithm for computing an O/PSE in SSSCGs with arbitrary costs where the leader is restricted to pure-strategy commitments.
It is based on a dynamic programming algorithm proposed in~\cite{ieong2005fast} for the computation of an optimal NE in symmetric SCGs without leadership.
The original algorithm runs in $O(n^6 r^5)$.
One can compute an O/PSE in $r$ iterations, fixing, at each iteration, the action the leader would choose and calling the previous algorithm to compute a NE which either minimizes or maximizes the leader's cost.
This takes, overall, $O(n^6 r^6)$.

We show, in the following, how to improve the lower the complexity of the original algorithm to $O(n^4 r^3)$, which allows for computing an O/PSE for the restricted case in $O(n^4 r^4)$.
The algorithm is based on the same recursive formula shown in~\cite{ieong2005fast}, which we reintroduce, here, in a different and, possibly, clearer way.

Let $A(h,B,M,V)$ be the cost of an optimal NE for a symmetric SCG without leadership restricted to $h$ resources $\{1,2,...,h\} \subseteq R$ and $B$ players, where $M$ is the largest cost incurred by a player and $V$ is the smallest cost a player would incur if she were to switch to another resource.

\begin{proposition}
$A(h,B,M,V)$ satisfies the following recursive equation:
\begin{align}
\label{dp:obj}
A(h,B,M,V) = \min_{\substack{p \in \{0, \dots, B\}\\m \in \mathbb{Z}^+, v \in \mathbb{Z}^+}} \quad
& A(h-1,p,m,v) + (B-p) \, c_h (B-p)\\
\label{dp:12}
\textnormal{s.t.} \quad
& m \leq M\\
\label{dp:13}
& v \geq V\\
\label{dp:2}
& c_h(B-p) \leq M\\
\label{dp:4}
& c_h(B-p+1) \geq V\\
\label{dp:3}
& c_h(B-p) \leq v\\
\label{dp:5}
& c_h(B-p+1) \geq m.
\end{align}
\end{proposition}
\begin{proof}
%
We show that all the 
constraints are necessary for the definition of $A(h,B,M,V)$ to be respected.
If Constraint~\eqref{dp:12} were not satisfied, $m > M$ would imply that there is at least a resource among those in $\{1, \dots, h-1\}$ costing strictly more than $M$.
If Constraint~\eqref{dp:13} were not satisfied, $v < V$ would imply that the cost to deviate to a resource among those in $\{1, \dots, h-1\}$ is strictly smaller than $V$.
If Constraint~\eqref{dp:2} were not satisfied, $c_h(B-p) > M$ would imply that $M$ is smaller than the cost of the most expensive chosen resource.
If Constraint~\eqref{dp:4} were not satisfied, $c_h(B-p+1) < V$ would imply that $V$ is larger than the cheapest cost a player would incur upon deviating to another resource.
If Constraint~\eqref{dp:3} were not satisfied, $c_h(B-p) > v$ would imply that each of the $B-p$ players who chose resource $h$ would have an incentive to deviate to any of the resources in $\{1, \dots, h-1\}$.
If Constraint~\eqref{dp:5} were not satisfied, $c_h(B-p+1) < m$ would imply that at least one of the $p$ players who selected a resource in $\{1, \dots, h-1\}$ (i.e., all those incurring a cost of $m$) would have an incentive to deviate to resource $h$.
\end{proof}

We show to simplify the recursive formula for $A(h,B,M,V)$:
\begin{theorem}\label{thm:step}
$A(h,B,M,V)$ satisfies the following recursive equation:
\begin{align}
A(h,B,M,V) = \min_{p \in \{0, \dots, B\}} \quad
& A(h-1,p,m(p)^*,v(p)^*) + (B-p) \, c_h(B-p)\\
\textnormal{s.t.} \quad
& c_h(B-p) \leq M\\
& c_h(B-p+1) \geq V,
\end{align}
where $m(p)^* = \min\{M, c_h(B-p+1)\}$ and $v(p)^* = \max\{V, c_h(B-p)\}$.
\end{theorem}
\begin{proof}
Constraints~\eqref{dp:12}--\eqref{dp:5} and~\eqref{dp:13}--\eqref{dp:3} imply, respectively, $m \leq \min\{M, c_h(B-p+1)\}$ and $v \geq \max\{V, c_h(B-p)\}$.
Hence, $m(p)^*$ and $v(p)^*$ are feasible for Problem~\eqref{dp:obj}--\eqref{dp:5}. 
Note that, if $m'>m$ and  $v'<v$, the feasible region underlying $A(h,p,m',v')$ contains the one underlying $A(h,p,m,v)$, which implies $A(h,p,m',v') \leq A(h,p,m,v)$.
The claim follows since $m(p)^*$ and $v(p)^*$ are, respectively, the largest and smallest values $m$ and $v$ can take.
\end{proof}

\begin{corollary}
In symmetric SCGs without leadership, an optimal NE can be found in $O(n^4 r^3)$. In SSSCGs with the leader restricted to pure strategies, an O/PSE can be found in $O(n^4 r^4)$.
\end{corollary}
\begin{proof}
Since there are at most $n r$ different values of $c_j(i)$, for all $j \in R$ and $i \in N$, there are at most $nr$ values of $M$ and at most $nr$ values of $V$.
There are also exactly $r$ values of $h$ and exactly $n$ of $B$.
Hence, the dynamic programming table of $A(h,B,M,V)$ contains $O(n^3 r^3)$ entries.
Due to Theorem~\eqref{thm:step}, computing an entry of the table requires $O(n)$.
Overall, an optimal NE is computed in $O(n^4 r^3)$.
For the case with leadership restricted to pure strategies, it suffices to run the algorithm for each resource the leader may choose, i.e., $O(r)$ times, obtaining a complexity of $O(n^4 r^4)$.
\end{proof}


\section{Mixed-Integer Linear Programming Formulations for Computing OSEs in Intractable SSCGs and SSSCGs}\label{sec:practical_algorithms}

In this section, we provide two MILP formulations for the problem of computing an OSE (in, in the worst case, exponential time) in SSCGs and SSSCGs for which the problem is intractable (see Sections~\ref{sec:complexity_different_actions}~and~\ref{sec:complexity_same_actions}).
Our goal is to provide methods which work suitably well in practice, even though their worst-case running time is exponential.\footnote{
We recall that, while we do not directly propose algorithms for the computation of PSEs for these intractable cases, their computation can be carried out with the general method proposed in~\cite{algorithmica} for general Stackelberg games in normal form.}

We start from SSSCGs, for which the MILP formulation is simpler, and then extend the result to the more general case of SSCGs. 

\subsection{Computing an OSE in SSSCGs (with arbitrary costs)}\label{sub_sec:milp_same_actions}

For the ease of notation, let $V = \{1, \ldots, n-1 \}$ be the set of possible congestion levels induced by the followers on a resource.
Let, for every resource $i \in R$ and value $v \in V$, the binary variable $y_{i v}$ be equal to 1 if and only if $\nu_i = v$, i.e., if and only if $v$ followers select resource $i \in R$. We use these variables to achieve a binarized representation of the followers' configuration $\nu \in \mathbb{N}^r$, namely, $\nu_i = \sum_{v \in V} v \, y_{i v}$ for all $i \in R$.
Let, for each $i \in R$, $\alpha_i \in [0,1]$ be equal to $\sigma_\ell(i)$.
Let also, for each $i \in R$ and $v \in V$, the auxiliary variable $z_{i v}$ be equal to the bilinear term $y_{i v} \alpha_i$.

The complete MILP formulation reads:
\begin{subequations}\label{eq:milp_same_actions}
	\begin{align}
		\min & \sum_{i \in R} \ \sum_{v \in V } c_{i, \ell}(v+1) \,  z_{i v} \label{eq:milp_same_actions_of} \\
   \textnormal{s.t.} & \sum_{v \in V } y_{iv} \leq 1 & \forall i \in R \label{eq:milp_same_actions_cons1} \\
		& \sum_{i \in R} \ \sum_{v \in V} v \, y_{i v}  = n-1 &  \label{eq:milp_same_actions_cons2}\\
		& \sum_{v \in V} \left(  y_{j v} c_{j,f}(v+1) + z_{j v} \Big( c_{j,f}(v+2) - c_{j,f}(v+1)  \Big)  \right) \geq \nonumber \hspace{-2cm}\\
                & \sum_{v \in V} \left(  y_{i v} c_{i,f}(v) + z_{i v} \Big( c_{i,f}(v+1) - c_{i,f}(v)  \Big)  \right) & \forall i \neq j \in R \label{eq:milp_same_actions_cons3} \\
		& z_{i v} \leq \alpha_i &  \forall i \in R, \forall v \in V \label{eq:milp_same_actions_cons4} \\
		& z_{i v} \leq y_{i v} & \forall i \in R, \forall v \in V \label{eq:milp_same_actions_cons5} \\
		& z_{i v} \geq \alpha_i + y_{i v} -1 & \forall i \in R, \forall v \in V \label{eq:milp_same_actions_cons6} \\
		& z_{i v} \geq 0 & \forall i \in R, \forall v \in V \label{eq:milp_same_actions_cons8}\\
		& \sum_{i \in R} \alpha_i = 1 \label{eq:milp_same_actions_cons7}\\
		& \alpha_i \geq 0 & i \in R \\
		& y_{i v} \in \{ 0,1 \} & \forall i \in R, \forall v \in V.
	\end{align}
\end{subequations}

Function~\eqref{eq:milp_same_actions_of} represents the leader's expected cost (to be minimized).
Constraints~\eqref{eq:milp_same_actions_cons1} ensure that at most one variable $y_{i v}$ be equal to $1$ for each resource $i \in R$, thus guaranteeing that the congestion level of each resource be uniquely determined (note that $\sum_{v \in V} y_{iv} = 0$ if no followers select resource $i \in R$).
Constraints~\eqref{eq:milp_same_actions_cons2} guarantee that the followers' configuration be well-defined, i.e., that $\sum_{i \in R} \nu_i$ be equal to $n-1$ (the number of followers).
Constraints~\eqref{eq:milp_same_actions_cons3} force the followers' configuration defined by the $y_{i v}$ variables to be an NE for the leader's strategy identified by the $\alpha_i$ variables.
This follows from the fact that $\displaystyle \sum_{v \in V} \left(  y_{i v} c_{i,f}(v) + z_{i v} \Big( c_{i,f}(v+1) - c_{i,f}(v)  \Big)  \right)$ (recall that $z_{iv} = y_{i v} \alpha_i $) is equal to the cost incurred by the followers who select resource $i \in R$, while $\displaystyle \sum_{v \in V} \left(  y_{j v} c_{j,f}(v+1) + z_{j v} \Big( c_{j,f}(v+2) - c_{j,f}(v+1)  \Big)  \right)$ (recall that $z_{j v} =  y_{j v} \alpha_j$) is equal to the cost they would incur after deviating to resource $j \in R$.
Let us remark that Constraints~\eqref{eq:milp_same_actions_cons3} are trivially satisfied if $y_{i v}= 0$ for all $v \in V$. This is correct as, if no followers choose resource $i \in R$, no equilibrium conditions need to be enforced.
Constraints~\eqref{eq:milp_same_actions_cons4}--\eqref{eq:milp_same_actions_cons8} are McCormick envelope constraints~\cite{mccormick1976computability} which guarantee $z_{i v} = y_{i v} \alpha_i$ whenever $y_{iv} \in \{0,1\}$.

We remark that Formulation~\eqref{eq:milp_same_actions} features $r (2n + 1)$ variables, $n r$ of which binary, and $r (r - 1) + r (3n + 1) + 2$ constraints.

\subsection{Computing an OSE in SSCGs}\label{sub_sec:milp_different_actions}

We now extend Formulation~\eqref{eq:milp_same_actions} to the case where the followers may have different action spaces, i.e., SSCGs.

For the ease of notation, let, for every $i \in R$, $\bar v_i = | \{ p \in F \mid i \in A_p \} |$ be the maximum number of followers who can select resource $i$, and let $V(i) = \{1, \ldots, \bar v_i \}$ be the set of possible congestion levels for resource $i$.
For every follower $p \in F$ and resource $i \in a_p$, let the binary variable $x_{p i}$ be equal to 1 if and only player $p$ selects resource $i$, i.e., if and only if $a_p = i$.
All the variables in~Formulation~\eqref{eq:milp_same_actions} are used with the same meaning.

The complete MILP formulation reads:
\begin{subequations}\label{eq:milp_different_actions}
	\begin{align}
		\min & \sum_{i \in R} \ \sum_{v \in V(i) } c_{i, \ell}(v+1) \,  z_{i v} \label{eq:milp_different_actions_of} \\
   \textnormal{s.t.} & \sum_{i \in A_p} x_{p i} = 1 & \forall p \in F \label{eq:milp_different_actions_cons0} \\
                & \sum_{v \in V(i) } y_{iv} \leq 1 & \forall i \in R \label{eq:milp_different_actions_cons1} \\
		&  \sum_{v \in V(i)} v \, y_{i v}  = \sum_{p \in F} x_{p i} & \forall i \in R  \label{eq:milp_different_actions_cons2}\\
		& \sum_{v \in V(i)} \left(  y_{j v} c_{j,f}(v+1) + z_{j v} \Big( c_{j,f}(v+2) - c_{j,f}(v+1)  \Big)  \right) \geq \nonumber \hspace{-5cm}\\
   & \geq \sum_{v \in V(i)} \left(  y_{i v} c_{i,f}(v) + z_{i v} \Big( c_{i,f}(v+1) - c_{i,f}(v)  \Big)  \right) \hspace{-.35cm} &  \forall p \in F, i \neq j \in A_p \label{eq:milp_different_actions_cons3} \\
		& z_{i v} \leq \alpha_i &  \forall i \in R, \forall v \in V(i) \label{eq:milp_different_actions_cons4} \\
		& z_{i v} \leq y_{i v} & \forall i \in R, \forall v \in V(i) \label{eq:milp_different_actions_cons5} \\
		& z_{i v} \geq \alpha_i + y_{i v} -1 & \forall i \in R, \forall v \in V(i) \label{eq:milp_different_actions_cons6} \\
		& z_{i v} \geq 0 & \forall i \in R, \forall v \in V(i) \label{eq:milp_different_actions_cons8}\\
		& \sum_{i \in R} \alpha_i = 1 \label{eq:milp_different_actions_cons7}\\
		& \alpha_i \geq 0 & i \in R \label{eq:milp_different_actions_cons77}\\
		& \alpha_i = 0 & i \in R \setminus A_\ell \label{eq:milp_different_actions_cons9} \\
                & x_{p i} \in \{ 0,1 \} & \forall p \in F, \forall i \in R \\
		& y_{i v} \in \{ 0,1 \} & \forall i \in R, \forall v \in V(i).
	\end{align}
\end{subequations}

Objective Function~\eqref{eq:milp_different_actions_of}, Constraints~\eqref{eq:milp_different_actions_cons1}, and Constraints~\eqref{eq:milp_different_actions_cons3}--\eqref{eq:milp_different_actions_cons77} have the same meaning as their counterparts in Formulation~\eqref{eq:milp_same_actions}. 
Constraints~\eqref{eq:milp_different_actions_cons0} ensure that each follower selects exactly one resource.
Constraints~\eqref{eq:milp_different_actions_cons2} guarantee that the followers' configuration be well-defined, i.e., that, for each $i \in R$, $\nu_i = \sum_{v \in V} v \, y_{iv}$ be equal to $\sum_{p \in F} x_{pi}$, i.e., to the number of followers who select resource $i$.
Notice that, differently from the previous formulation, Constraints~\eqref{eq:milp_different_actions_cons3} are enforced for each follower $p \in F$ here, and only for pairs of resources $i, j \in R$ follower $p$ has access to.
Note also that, via Constraints~\eqref{eq:milp_different_actions_cons9}, $\alpha_i$ is forced to be equal to 0 for all the resources $i \in R$ the leader has no access to.

We observe that Formulation~\eqref{eq:milp_different_actions} features $\sum_{p \in F} |A_p| + 2 \sum_{i \in R} \bar{v}_i + r = O(r(3n+1))$ variables, $\sum_{p \in F} |A_p| + \sum_{i \in R} \bar{v}_i = O(2rn)$ of which binary, and $n+2r+3 \sum_{i \in R} \bar{v}_i  + \sum_{p \in F} |A_p| \left( |A_p| - 1 \right) = O(n+2r+3nr+nr(r-1))$ constraints.

\subsection{Experimental Evaluation}\label{subsec:experimental}

While the scalability of all the efficient algorithms we proposed in Section~\ref{sec:algorithms_same_actions_polynomial} is clear due to their polynomiality, it is of interest to assess, experimentally, how state-of-the-art branch-and-bound methods behave when solving our formulations for the intractable cases on game instances of increasing size.

For the purpose, we experiment with two MILP formulations we proposed
on a testbed of randomly generated game instances of two classes:
\begin{itemize}
\item \emph{SSSCG instances}: we assume a number of followers
in $\{20,40,60,80,100\}$, with $r$ resources in the range  $\{10,20,30,40,50\}$ and players' costs randomly generated by sampling from $\{1,\ldots,(n-1) \, r\}$ with a uniform probability.\footnote{The value $(n-1) \, r$ is chosen as, when looking for pure-strategy NEs, cost functions taking $(n-1) \, r$ different values are sufficient to represent every possible singleton congestion game.}
%
\item \emph{SSCG instances}: we assume a number of followers
in $\{20,40,60,80,100\}$, with $r = 30$ resources and a number of actions $|A_p|$ per player in the range $\{7, 15, 22\}$, generated by sampling without replacement; the players' costs are sampled from  $\{1,\ldots,(n-1) \, r\}$ with uniform probability.
\end{itemize}
We generate 15 instances per combination of the parameters.
All the experiments are run on a UNIX machine with a total of 32 cores working at 2.3 GHz, equipped with 128 GB of RAM.
Each game instance is solved on a single core within a time limit of 7200 seconds.
We use Python 2.7, solving the MILP formulations with GUROBI 7.0.


We use, as baseline for the comparisons, a simple algorithm which, starting from a randomly generated assignment of players to the resources,
simulates a best-response dynamics, halting after a time limit of 10 minutes.
When ties arise, i.e., whenever the are more than a single player who are not playing their best response, we select a player lexicographically and make her switch to playing her (currently) best response.
We refer to this algorithm as a {\em best response dynamics heuristic} as the method is not exact when applied to the intractable cases of SSCGs and SSSCGs.
On average, within the time limit of 10 minutes we observe a number of deviations to a best response of the order of $10^5$.
Let us recall that the method always produces, by design, pure-strategy NEs.




%
Figure~\ref{fig:gull12}~(a) and~(b) report the results for SSSCGs with arbitrary costs with 30 resources.
Figure~\ref{fig:gull12}~(a) displays the average computing time required by MILP Formulation~\eqref{eq:milp_same_actions}, as a function of the number of followers and for a different number of actions available to each player.
One can see that, with Formulation~\eqref{eq:milp_same_actions}, an optimal solution is always found within the time limit of 7200 seconds in all the instances.
This suggests that, even if the problem is hard in the worst case, an optimal solution can be found in a reasonable amount of time on randomly generated instances.
%
Figure~\ref{fig:gull12}~(b) reports, as a function of the number of followers, the average leaders' cost of the solutions obtained with MILP Formulation~\eqref{eq:milp_same_actions}, compared to the average cost obtained with the best response dynamics heuristic.
As the figure shows, the difference in leader's utility between solutions found with the two methods can be quite large as the number of followers increases, up to a factor of 6 with $n = 100$, showing a clearly growing trend.

\begin{figure}[h!]
  \centering
  \begin{subfigure}[b]{0.49\textwidth}
    \includegraphics[width=1.1\textwidth]{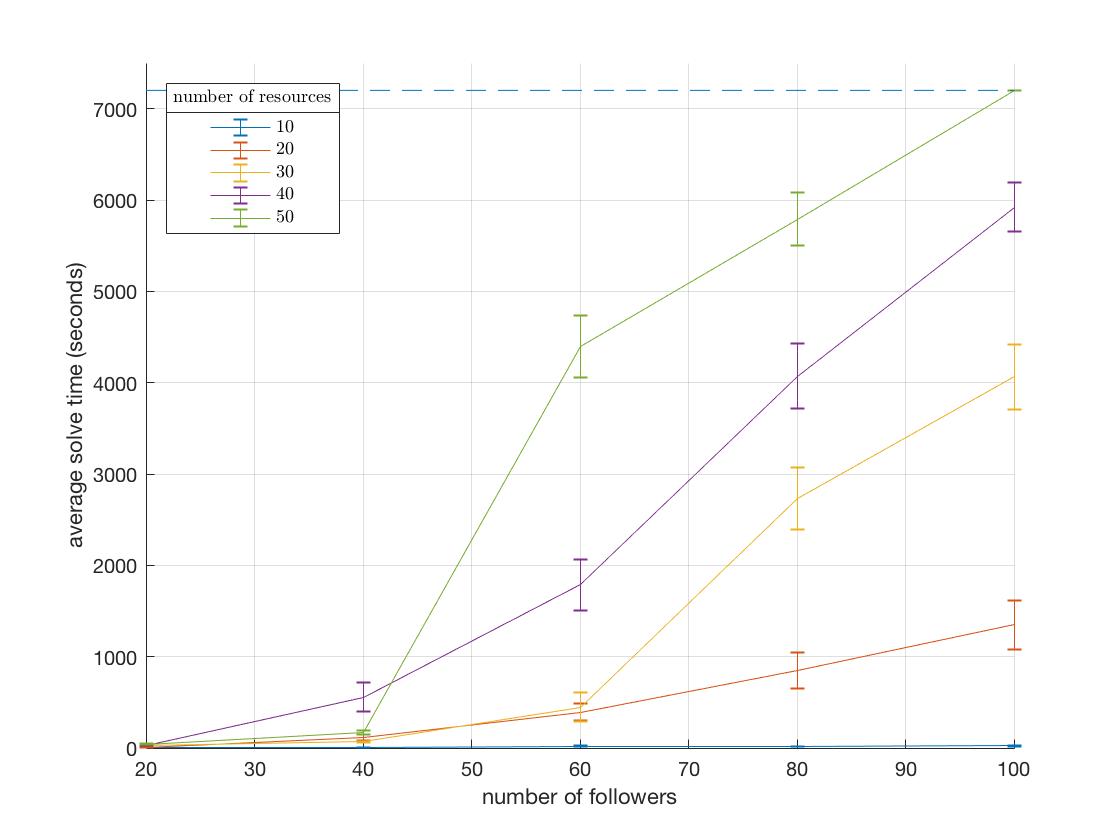}
    \caption{ }
  \end{subfigure}
  \begin{subfigure}[b]{0.49\textwidth}
    \includegraphics[width=1.1\textwidth]{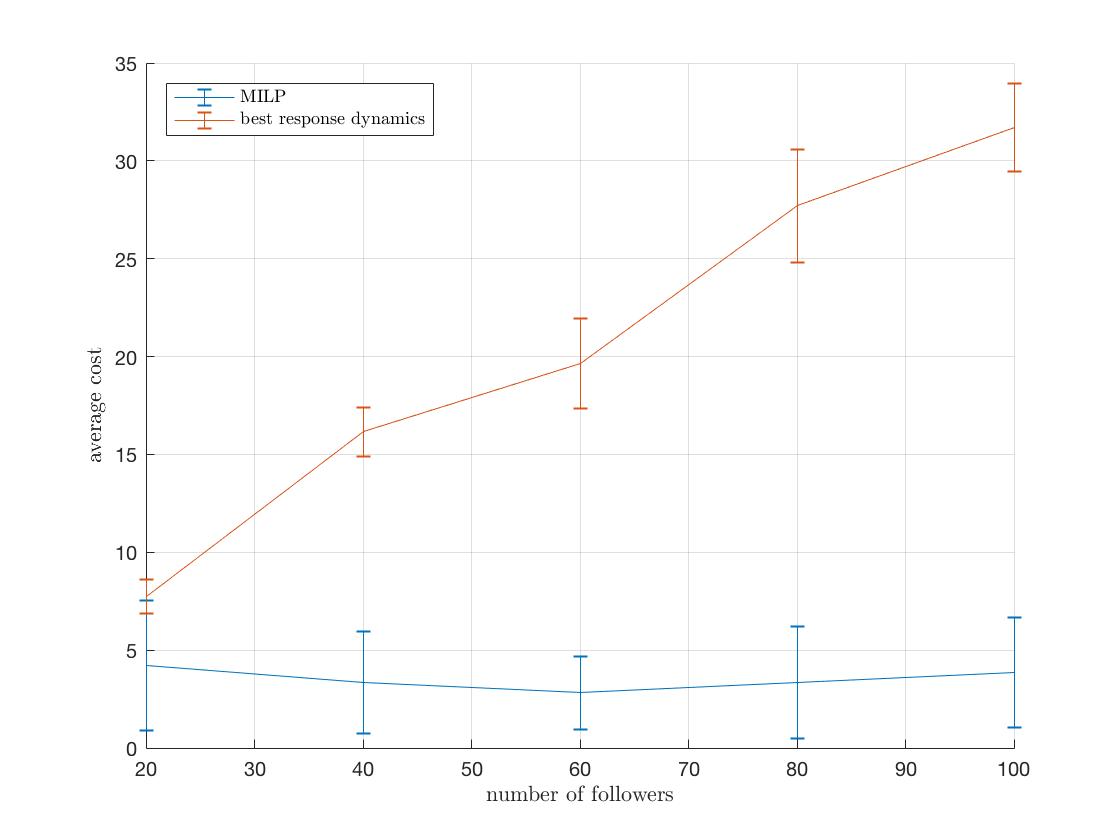}
    \caption{ }
  \end{subfigure}
  \caption{
    Results for the computation of an OSE in SSSCGs with arbitrary costs 30 resources.
    (a) Average computing time required by Formulation~\eqref{eq:milp_same_actions}, as a function of the number of followers and for a different number of actions available to each player.
    (b) Average leaders' cost of the solutions obtained with MILP Formulation~\eqref{eq:milp_same_actions} and with the best response dynamics heuristic as a function of the number of followers, with 15 actions per player.}
  \label{fig:gull12}
\end{figure}        

Figure~\ref{fig:gull34}~(a) and~(b) report the results for SSCGs with arbitrary costs with 30 resources.
Figure~\ref{fig:gull34}~(a) reports the average computing time required by Formulation~\eqref{eq:milp_different_actions} to find an OSE, as a function of the number of followers and for a different number of actions available to each player.
Similarly to the case of SSSCGs, the chart shows that with Formulation~\eqref{eq:milp_different_actions} we can find an optimal solution within the time limit of 7200 seconds in all the instances.
This suggests that, even if the problem is hard in the worst case, also for SSGGs one can find an optimal solution in a reasonable amount of computing time on randomly generated instances.
The chart also shows, though, that the time required to solve this class of problems is much larger than the time required to solve their SSSCGs counterparts.
Figure~\ref{fig:gull34}~(b) reports, for games with 15 actions per player, the average leader's cost of the solutions obtained with the MILP Formulation~\eqref{eq:milp_different_actions} and with the best response dynamics heuristic, as a function of the number of followers.
%
Differently from the case of SSSCGs, we observe that for SSCGs the heuristic
returns solutions which, empirically, appear to be within a constant approximation factor of the optimal ones, never larger than 5.

Overall, the results suggest the practical viability of our MILP formulations for finding provably optimal solutions also for games where a simple best response heuristic provide poor-quality solutions.



\begin{figure}[h!]
  \centering
  \begin{subfigure}[b]{0.49\textwidth}
    \includegraphics[width=1.1\textwidth]{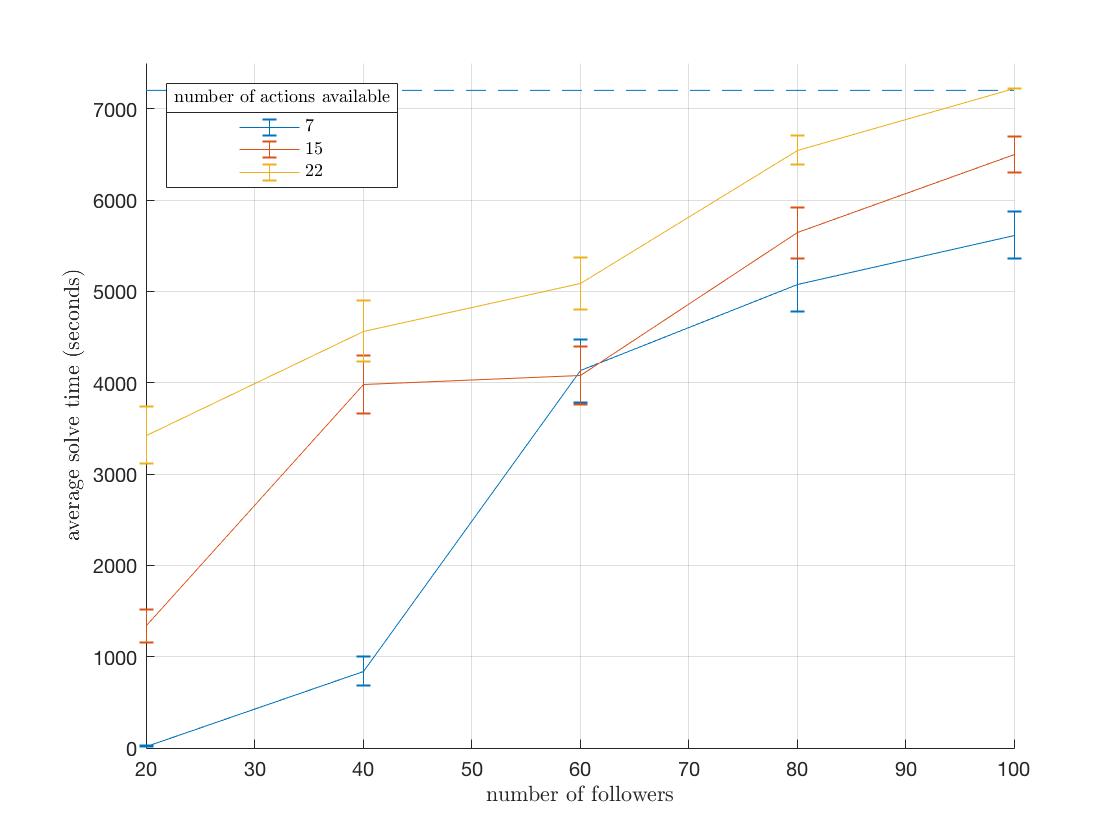}
    \caption{ }
  \end{subfigure}
  \begin{subfigure}[b]{0.49\textwidth}
    \includegraphics[width=1.1\textwidth]{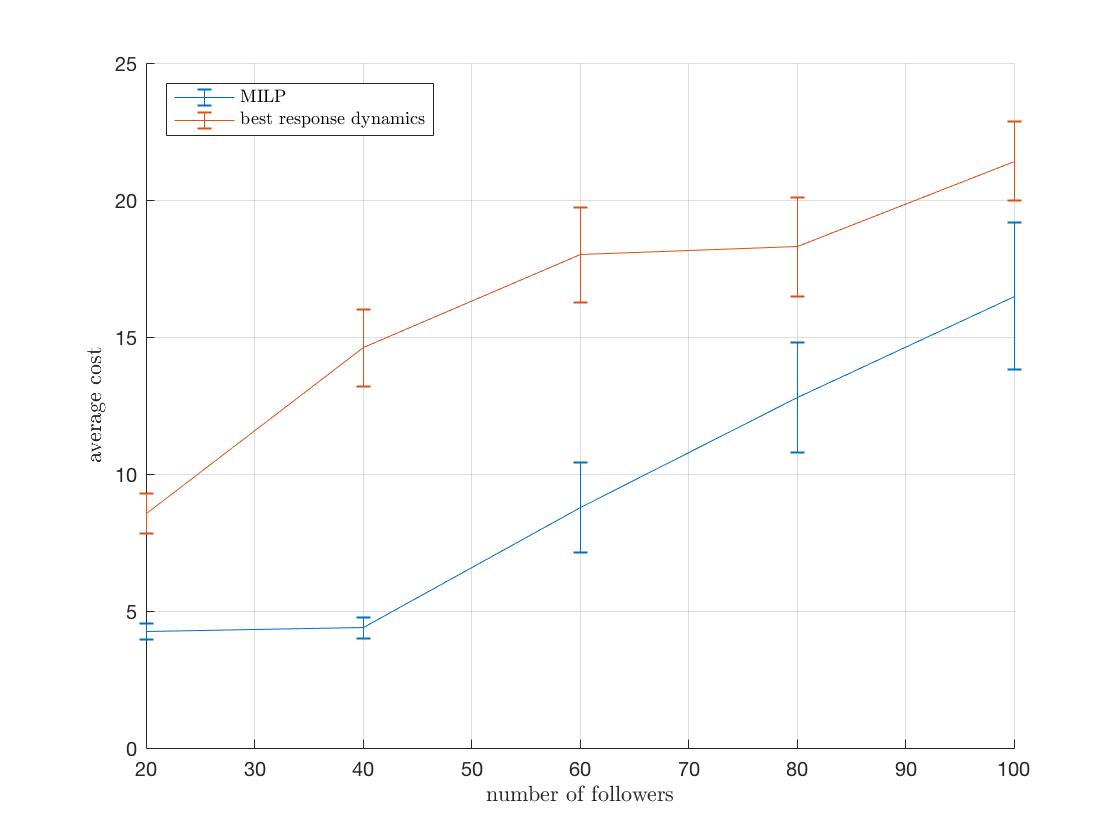}
    \caption{ }
  \end{subfigure}
    \caption{
      Results for the computation of an OSE in SSCGs with arbitrary costs and 30 resources.
      (a) Average computing time required by Formulation~\eqref{eq:milp_different_actions}, as a function of the number of followers and for a different number of actions available to each player.
      (b) Average leader's cost of the solutions obtained with the MILP Formulation~\eqref{eq:milp_different_actions} and with the best response dynamics heuristics, as a function of the number of followers and with 15 actions per player.}
    \label{fig:gull34}
\end{figure}

\section{Conclusions and Future Works}\label{sec:conclusions}

We have analyzed Stackelberg games where the underlying structure is a congestion game, focusing on the case in which the players' actions are singletons.

We have shown that the problem of computing a Stackelberg Equilibrium (SE) in such games is hard, except for the case in which all the players share the same resources and the cost functions are monotonically increasing in the congestion level.
More precisely, we have shown that, for games where either the players have different action spaces and their cost functions are monotonic, or their action spaces are the same but their cost functions are nonmonotonic, it is not possible to approximate in polynomial time the leader's cost at an either optimistic SE (OSE) or pessimistic SE (PSE) up to within any factor polynomial in the size of the game, unless $\mathsf{P}=\mathsf{NP}$.

We have proposed a polynomial-time algorithm for finding an O/PSE for the case where the players have the same action spaces and their cost functions are monotonic, and we have shown that games in this class always admit a pure-strategy SE.
We have also shown how to improve the complexity of the state-of-the-art algorithm for the computation of an optimal NE in singleton congestion games, which has allowed us to compute an O/PSE in polynomial time for the case where the leader is restricted to pure strategies.
For the intractable cases with different action spaces and generic cost functions, we have proposed a mixed-integer linear programming formulation for finding an OSE, and a more compact one for the case in which the space of actions are the same.
We have shown that state-of-the-art solvers scale well in practice when solving our formulations on random game instances, allowing for tackling games with up to 40 resources and 100 followers.
The experiments have also revealed that a simple heuristic algorithm based on the repetition of best-response dynamics
returns high-quality solutions.

In the future, we will investigate whether congestion games with a special structure allow for efficient solution algorithms.
We will also investigate whether the introduction of more than a single leader makes the problem harder and, finally, we will study practical applications of the methods we developed to, e.g., routing and queuing problems.


\section*{References}

\bibliography{refs}

\begin{thebibliography}{10}
\expandafter\ifx\csname url\endcsname\relax
  \def\url#1{\texttt{#1}}\fi
\expandafter\ifx\csname urlprefix\endcsname\relax\def\urlprefix{URL }\fi
\expandafter\ifx\csname href\endcsname\relax
  \def\href#1#2{#2} \def\path#1{#1}\fi

\bibitem{Marchesi18:leadership}
A.~Marchesi, S.~Coniglio, N.~Gatti, Leadership in singleton congestion games,
  in: IJCAI, 2018.

\bibitem{von2010leadership}
B.~Von~Stengel, S.~Zamir, Leadership games with convex strategy sets, Games and
  Economic Behavior 69~(2) (2010) 446--457.

\bibitem{conitzer2006computing}
V.~Conitzer, T.~Sandholm, Computing the optimal strategy to commit to, in:
  Proceedings of the 7th ACM conference on Electronic commerce, 2006, pp.
  82--90.

\bibitem{paruchuri2008playing}
P.~Paruchuri, J.~P. Pearce, J.~Marecki, M.~Tambe, F.~Ordonez, S.~Kraus, Playing
  games for security: an efficient exact algorithm for solving bayesian
  stackelberg games, in: AAMAS, 2008, pp. 895--902.

\bibitem{KiekintveldJTPOT09}
C.~Kiekintveld, M.~Jain, J.~Tsai, J.~Pita, F.~Ord{\'{o}}{\~{n}}ez, M.~Tambe,
  Computing optimal randomized resource allocations for massive security games,
  in: AAMAS, 2009, pp. 689--696.

\bibitem{an2011guards}
B.~An, J.~Pita, E.~Shieh, M.~Tambe, C.~Kiekintveld, J.~Marecki, Guards and
  {Protect}: Next generation applications of security games, ACM SIGecom
  Exchanges 10~(1) (2011) 31--34.

\bibitem{tambe2011security}
M.~Tambe, Security and Game Theory: Algorithms, Deployed Systems, Lessons
  Learned, Cambridge University Press, 2011.

\bibitem{labbe1998bilevel}
M.~Labb{\'e}, P.~Marcotte, G.~Savard, A bilevel model of taxation and its
  application to optimal highway pricing, Management science 44~(12-part-1)
  (1998) 1608--1622.

\bibitem{labbe2016bilevel}
M.~Labb{\'e}, A.~Violin, Bilevel programming and price setting problems, ANN
  OPER RES 240~(1) (2016) 141--169.

\bibitem{caprara2016bilevel}
A.~Caprara, M.~Carvalho, A.~Lodi, G.~J. Woeginger, Bilevel knapsack with
  interdiction constraints, INFORMS J COMPUT 28~(2) (2016) 319--333.

\bibitem{matuschke2017protection}
J.~Matuschke, S.~T. McCormick, G.~Oriolo, B.~Peis, M.~Skutella, Protection of
  flows under targeted attacks, OPER RES LETT 45~(1) (2017) 53--59.

\bibitem{amaldi2013network}
E.~Amaldi, A.~Capone, S.~Coniglio, L.~G. Gianoli, Network optimization problems
  subject to max-min fair flow allocation, IEEE COMMUN LETT 17~(7) (2013)
  1463--1466.

\bibitem{sandholm2002evolutionary}
W.~H. Sandholm, Evolutionary implementation and congestion pricing, The Review
  of Economic Studies 69~(3) (2002) 667--689.

\bibitem{Breton88:Sequential}
M.~Breton, A.~Alj, A.~Haurie, Sequential {S}tackelberg equilibria in two-person
  games, Journal of Optimization Theory and Applications 59~(1) (1988) 71--97.

\bibitem{conitzer2011commitment}
V.~Conitzer, D.~Korzhyk, Commitment to correlated strategies, in: AAAI, 2011,
  pp. 632--637.

\bibitem{basilico2017methods}
N.~Basilico, S.~Coniglio, N.~Gatti, Methods for finding leader-follower
  equilibria with multiple followers, CoRR abs/1707.02174.
\newblock \href {http://arxiv.org/abs/1707.02174} {\path{arXiv:1707.02174}}.

\bibitem{algo2018computing}
S.~Coniglio, N.~Gatti, A.~Marchesi, Computing a pessimistic leader-follower
  equilibrium with multiple followers: the mixed-pure case, CoRR
  abs/1808.01438.
\newblock \href {http://arxiv.org/abs/1808.01438} {\path{arXiv:1808.01438}}.

\bibitem{coniglio2017pessimistic}
S.~Coniglio, N.~Gatti, A.~Marchesi, Pessimistic leader-follower equilibria with
  multiple followers, in: IJCAI, 2017, pp. 171--177.

\bibitem{basilico2017bilevel}
N.~Basilico, S.~Coniglio, N.~Gatti, A.~Marchesi, Bilevel programming approaches
  to the computation of optimistic and pessimistic single-leader-multi-follower
  equilibria, LEIBNIZ INTERNATIONAL PROCEEDINGS IN INFORMATICS 75 (2017) 1--14.

\bibitem{letchford2009learning}
J.~Letchford, V.~Conitzer, K.~Munagala, Learning and approximating the optimal
  strategy to commit to, in: International Symposium on Algorithmic Game
  Theory, Springer, 2009, pp. 250--262.

\bibitem{alberto2018computing}
G.~{De~Nittis}, A.~Marchesi, N.~Gatti, Computing the optimal strategy to commit
  to in polymatrix games, in: AAAI, 2018, pp. 82--90.

\bibitem{Letchford10:Computing}
J.~Letchford, V.~Conitzer, Computing optimal strategies to commit to in
  extensive-form games, in: EC, 2010.

\bibitem{Farina18:trembling}
G.~Farina, A.~Marchesi, C.~Kroer, N.~Gatti, T.~Sandholm, Trembling-hand
  perfection in extensive-form games with commitment, in: IJCAI, 2018.

\bibitem{Bosansky15:Sequence}
B.~Bo{\v{s}}ansk{\`y}, J.~Cermak, Sequence-form algorithm for computing
  stackelberg equilibria in extensive-form games, in: AAAI, 2015.

\bibitem{Cermak16:Using}
J.~Cermak, B.~Bo{\v{s}}ansk{\`y}, K.~Durkota, V.~Lisy, C.~Kiekintveld, Using
  correlated strategies for computing {S}tackelberg equilibria in
  extensive-form games, in: AAAI, 2016.

\bibitem{Kroer18:Robust}
C.~Kroer, G.~Farina, T.~Sandholm, Robust {Stackelberg} equilibria in
  extensive-form games and extension to limited lookahead, in: AAAI, 2018.

\bibitem{de2018computing}
G.~De~Nittis, A.~Marchesi, N.~Gatti, Computing the strategy to commit to in
  polymatrix games {(Extended Version)}, CoRR abs/1807.11914.
\newblock \href {http://arxiv.org/abs/1807.11914} {\path{arXiv:1807.11914}}.

\bibitem{letchford2012computing}
J.~Letchford, L.~MacDermed, V.~Conitzer, R.~Parr, C.~L. Isbell, Computing
  optimal strategies to commit to in stochastic games., in: AAAI, 2012.

\bibitem{vorobeychik2012computing}
Y.~Vorobeychik, S.~P. Singh, Computing stackelberg equilibria in discounted
  stochastic games., in: AAAI, 2012.

\bibitem{xu2016signaling}
H.~Xu, R.~Freeman, V.~Conitzer, S.~Dughmi, M.~Tambe, Signaling in bayesian
  stackelberg games, in: Proceedings of the 2016 International Conference on
  Autonomous Agents \& Multiagent Systems, International Foundation for
  Autonomous Agents and Multiagent Systems, 2016, pp. 150--158.

\bibitem{rosenthal1973class}
R.~W. Rosenthal, A class of games possessing pure-strategy nash equilibria,
  International Journal of Game Theory 2~(1) (1973) 65--67.

\bibitem{monderer1996potential}
D.~Monderer, L.~S. Shapley, Potential games, Games and economic behavior 14~(1)
  (1996) 124--143.

\bibitem{ieong2005fast}
S.~Ieong, R.~McGrew, E.~Nudelman, Y.~Shoham, Q.~Sun, Fast and compact: A simple
  class of congestion games, in: AAAI, 2005, pp. 489--494.

\bibitem{ackermann2008impact}
H.~Ackermann, H.~R{\"o}glin, B.~V{\"o}cking, On the impact of combinatorial
  structure on congestion games, Journal of the ACM (JACM) 55~(6) (2008) 25.

\bibitem{fabrikant2004complexity}
A.~Fabrikant, C.~Papadimitriou, K.~Talwar, The complexity of pure nash
  equilibria, in: Proceedings of the thirty-sixth annual ACM symposium on
  Theory of computing, ACM, 2004, pp. 604--612.

\bibitem{werneck2000finding}
R.~Werneck, J.~Setubal, A.~da~Conceicao, Finding minimum congestion spanning
  trees, Journal of Experimental Algorithmics (JEA) 5 (2000) 11.

\bibitem{roughgarden2004stackelberg}
T.~Roughgarden, Stackelberg scheduling strategies, SIAM Journal on Computing
  33~(2) (2004) 332--350.

\bibitem{fotakis2010stackelberg}
D.~Fotakis, Stackelberg strategies for atomic congestion games, Theory of
  Computing Systems 47~(1) (2010) 218--249.

\bibitem{bonifaci2010stackelberg}
V.~Bonifaci, T.~Harks, G.~Sch{\"a}fer, Stackelberg routing in arbitrary
  networks, Mathematics of Operations Research 35~(2) (2010) 330--346.

\bibitem{bilo2015stackelberg}
V.~Bil{\`o}, C.~Vinci, On stackelberg strategies in affine congestion games,
  in: International Conference on Web and Internet Economics, Springer, 2015,
  pp. 132--145.

\bibitem{shoham2008multiagent}
Y.~Shoham, K.~Leyton-Brown, Multiagent systems: Algorithmic, game-theoretic,
  and logical foundations, Cambridge University Press, 2008.

\bibitem{garey1979computers}
M.~R. Garey, D.~S. Johnson, Computers and Intractability: A Guide to the Theory
  of {NP}-completeness, WH Freeman and Company, 1979.

\bibitem{mccormick1976computability}
G.~McCormick, {Computability of global solutions to factorable nonconvex
  programs: Part I -- Convex underestimating problems}, Math. Program. 10~(1)
  (1976) 147--175.

\end{thebibliography}

\end{document}